\documentclass[10pt,journal,a4paper]{IEEEtran} %TWC

\usepackage{amssymb}
\usepackage{amsmath}
\usepackage{amsmath,bm}
\usepackage{amsthm}
\usepackage{graphicx}
\usepackage{cuted}
\usepackage{subfigure}
\usepackage{multirow}
\usepackage{makecell} 
\usepackage{cite}
\usepackage{enumerate}
\usepackage{enumitem}
\usepackage{stfloats}
\usepackage[ruled]{algorithm2e}
\usepackage{algorithmic}  
\usepackage{color, soul}
\usepackage[labelsep=period]{caption}

\DeclareMathOperator{\sinc}{sinc}
\allowdisplaybreaks[4]
\begin{document}
	\newtheorem{theorem}{\bf~~Theorem}
	\newtheorem{remark}{\bf~~Remark}
	\newtheorem{observation}{\bf~~Observation}
	\newtheorem{definition}{\bf~~Definition}
	\newtheorem{lemma}{\bf~~Lemma}
	\newtheorem{preliminary}{\bf~~Preliminary}
	\newtheorem{proposition}{\bf~~Proposition}
	\newtheorem{comment}{\bf~~Comment}
	\renewcommand\arraystretch{0.9}
%	\title{Reconfigurable Refractive Surface-Enabled Multi-User Holographic MIMO Communications}
	\title{\huge{Dual-Polarized Reconfigurable Intelligent Surface-Based Antenna for Holographic MIMO Communications}}
%	 \title{\setlength{\baselineskip}{3pt}\Large Title}
%	 %\vspace{0.2cm}
%\title{{ Intelligent Omni-Surfaces: Reflection-Refraction Model, Full-dimensional Beamforming, and System Implementation}}
%\title{{ Intelligent Omni-Surfaces: Reflection-Refraction Model, Full-dimensional Beamforming, and System Implementation}}
	\author{\normalsize \IEEEauthorblockN{
			{Shuhao Zeng}, \IEEEmembership{\normalsize Member, IEEE},
			{Hongliang Zhang}, \IEEEmembership{\normalsize Member, IEEE},
			{Boya Di}, \IEEEmembership{\normalsize Member, IEEE},\\
			{Zhu Han}, \IEEEmembership{\normalsize Fellow, IEEE},
			{and H. Vincent Poor}, \IEEEmembership{\normalsize Fellow, IEEE}}
%		\thanks{Manuscript received November 21, 2022; revised June 29, 2023; accepted September 24, 2023. This work was supported in part by the National Key R\&D Project of China under Grant 2022YFE0111900, the National Science Foundation under Grants 62271012 and 62227809, and the Beijing Natural Science Foundation under Grants L212027 and 4222005. The associate editor coordinating the review of this paper and approving it for publication was Alessio Zappone. \emph{(Corresponding author: Lingyang Song.)}}
%		This work was supported in part by the National Key R\&D Project of China under Grant No. 2020YFB1807100, National Natural Science Foundation Grant 62271012, National Science Foundation under Grant 61931019 and 61941101, Beijing Natural Science Foundation under Grant L212027 and 4222005, Peking University Research Grant 7100603967, in part by the Engineering and Physical Sciences Research Council (EPSRC) under Project EP/W035588/1, in part by the European Commission through the H2020 ARIADNE Project under Grant 871464, in part by the H2020 RISE-6G Project under Grant 101017011, in part by NSF CNS-2107216 and CNS-2128368, and in part by U.S. National Science Foundation under Grant CNS-2128448.
		\thanks{Shuhao Zeng is with School of Electronic and Computer Engineering, Peking University Shenzhen Graduate School, Shenzhen, China, and also with Department of Electrical and Computer Engineering, Princeton University, NJ, USA (email: shuhao.zeng96@gmail.com).}
		\thanks{Hongliang Zhang and Boya Di are with School of Electronics, Peking University, Beijing 100871, China (email: hongliang.zhang@pku.edu.cn; boya.di@pku.edu.cn).}
		\thanks{Zhu Han is with Electrical and Computer Engineering Department, University of Houston, Houston, TX, USA, and also with the Department of Computer Science and Engineering, Kyung Hee University, Seoul, South Korea (email: hanzhu22@gmail.com).}		
		\thanks{H. Vincent Poor is with Department of Electrical and Computer Engineering, Princeton University, NJ, USA (email: poor@princeton.edu).}
%		\thanks{H. Zhang is with Department of Electrical and Computer Engineering, Princeton University, Princeton, NJ 08544, USA (email: hz16@princeton.edu, poor@princeton.edu).}
%		\thanks{H. Zhang and H. V. Poor are with Department of Electrical and Computer Engineering, Princeton University, Princeton, NJ 08544, USA (email: hz16@princeton.edu, poor@princeton.edu).}
%		\thanks{Y. Liu is with the School of Electronic Engineering and Computer Science, Queen Mary University of London, London E1 4NS, U.K. (e-mail: yuanwei.liu@qmul.ac.uk).}
%		\thanks{M. Di Renzo is with Universit\'{e} Paris-Saclay, CNRS, CentraleSup\'{e}lec, Laboratoire des Signaux et Syst\`{e}mes, 3 Rue Joliot-Curie, 91192 Gif-sur-Yvette, France. (marco.di-renzo@universite-paris-saclay.fr)}
%		\thanks{Z. Han is with Electrical and Computer Engineering Department, University of Houston, Houston, TX, USA, and also with the Department of Computer Science and Engineering, Kyung Hee University, Seoul, South Korea (email: zhan2@uh.edu).}
	}
	
	\maketitle
	%%\vspace{-2cm}
	\begin{abstract}
	%\vspace{-0.3cm}
	Holographic multiple-input-multiple output (HMIMO), which is enabled by large-scale antenna arrays with quasi-continuous apertures, is expected to be an important technology in the forthcoming 6G wireless network. Reconfigurable intelligent surface~(RIS)-based antennas provide an energy-efficient solution for implementing HMIMO. Most existing works in this area focus on single-polarized RIS-enabled HMIMO, where the RIS can only reflect signals in one polarization towards users and signals in the other polarization cannot be received by intended users, leading to degraded data rate. To improve multiplexing performance, in this paper, we consider a dual-polarized RIS-enabled single-user HMIMO network, aiming to optimize power allocations across polarizations and analyze corresponding maximum system capacity. However, due to {interference between different polarizations}, the dual-polarized system cannot be simply decomposed into two independent single-polarized ones. Therefore, existing methods developed for the single-polarized system cannot be directly applied, which makes the optimization and analysis of the dual-polarized system challenging. To cope with this issue, we derive an asymptotically tight upper bound on the ergodic capacity, based on which the power allocations across two polarizations are optimized. Potential gains achievable with such dual-polarized RIS are analyzed. Numerical results verify our analysis.

	\end{abstract}
%\vspace{-0.3cm}
	\begin{IEEEkeywords}
	Dual-polarized reconfigurable intelligent surface, holographic MIMO, 6G, multiplexing 
		%\vspace{-0.4cm}
	
	\end{IEEEkeywords}
%\vspace{-.8cm}
%\newpage
		
%	\printnomenclature
	\section{Introduction}
	To meet the stringent data rate requirement of the future sixth generation~(6G) wireless networks, holographic multiple-input-multiple-output~(HMIMO) is a key conceptual enabler, {where a large number of tiny elements are integrated into a compact space~\cite{A_HMIMO_2020,Gong_holographic_2024}. Benefited from the large radiation aperture of the resulting antenna array, HMIMO can achieve high directive gain~\cite{Balanis_book_antenna_theory_2016,Ji_Extra_2023}}, and thus is capable of supporting high-speed data transmissions. However, limited by practical power budgets, it is difficult to implement HMIMO with a conventional phased array since it requires numerous energy-intensive phase shifters, leading to unacceptable power consumption~\cite{Zeng_holographic_multi_TWC_2023}. Unlike phased arrays, reconfigurable intelligent surface~(RIS)-based antennas serve as more energy efficient enablers of HMIMO~\cite{Zeng_holographic_multi_TWC_2023}. In particular, an RIS consists of an array of sub-wavelength elements~\cite{Zeng_coverage_2021}, which can reflect incident electromagnetic~(EM) signals and apply adjustable phase shifts~\cite{Palma_transmit_array_2015}. Different from a phased array, RIS elements realize phase tunability through ultra-low-power diodes without the need of phase shifters~\cite{Zeng_both_2021}. Therefore, by utilizing RIS-based antennas, the base station~(BS) can generate highly directive beams towards users with significantly reduced power consumption compared with conventional phased arrays.
	
	Existing works on RIS-enabled HMIMO mainly focus on single-polarized RIS~\cite{Zeng_holographic_multi_TWC_2023,Han_mmWave_2020}. More specifically, wireless signals are essentially a kind of electromagnetic~(EM) waves, which exhibit polarization characteristics, i.e., {the direction of electric field strength is time-variant~\cite{Wei_Tri_2023}}. Since the electric field strength can only oscillate in the plane vertical to the propagation direction of the EM waves, there are two orthogonal polarizations, e.g., horizontal polarization and vertical polarization. The single-polarized RIS considered in existing works can only reflect signals in one polarization towards users and signals in the other polarization cannot be received by the users, leading to degraded communication performances. 
	
	To fully exploit the polarization dimension, in this paper, we consider a dual-polarized RIS-aided single-user HMIMO communication system. By deploying a dual-polarized RIS at the BS as the transmit antenna, the number of antenna elements can be doubled in the same physical enclosure, since each dual-polarized RIS element consists of two co-located single-polarized elements~\cite{Ramezani_dual_polarized_2023}, as shown in Fig.~\ref{sysmodel}. Besides, benefiting from the capability of reflecting signals in two polarization simultaneously while applying independently controllable phase shifts~\cite{Qian_Low_complexity_2022}, a dual-polarized RIS can improve multiplexing performance. In such a dual-polarized system, our goal is to optimize power allocations across polarizations and analyze the corresponding maximum system capacity, which, however, is challenging. This is because dual-polarized signals suffer from cross-polarization component during propagation~\cite{Nabar_polarization_2002}, leading to {interference between different polarizations}. Therefore, the dual-polarized system cannot be simply decomposed into two independent single-polarized ones, which indicates that existing optimization and analysis frameworks originally designed for the single-polarized system cannot be directly applied.

While dual-polarized multi-antenna communications have been extensively studied, the focus has predominantly been on conventional massive MIMO equipped with dual-polarized antennas~\cite{Sena_MIMO_NOMA_2019,Qian_tensor_2018,Gong_secure_communication_2017,Ozdogan_dual_polarized_2023} and dual-polarized RIS-based relay~\cite{yu_dual_polarized_2022,Bhowal_D2D_RIS_2023}. For example, in~\cite{Sena_MIMO_NOMA_2019}, the authors consider a multi-cluster multi-user massive MIMO network with non-orthogonal multiple access (NOMA) and a dual-polarized antenna array, where the outage probability, diversity gain, and outage sum-rate are derived. In~\cite{Ozdogan_dual_polarized_2023}, the spectral efficiency of a dual-polarized massive MIMO network is maximized by optimizing the power allocation over different users and polarizations. The authors in~\cite{yu_dual_polarized_2022} consider a multi-user MIMO network assisted by a dual-polarized RIS-based relay, where the ergodic capacity is analyzed and the RIS phase shifts in two polarizations are jointly optimized to improve the capacity. 

\begin{figure}[!t]
	\centering
	\includegraphics[width=0.45\textwidth]{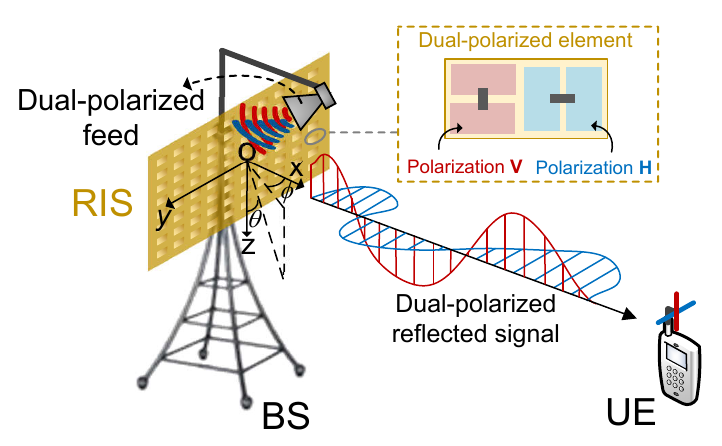}
			\vspace{1.5mm}
	\caption{System model of a dual-polarized downlink single-user network, where an RIS is deployed as the BS antenna.}
	%		\vspace{-5mm}
	\label{sysmodel}
\end{figure}

However, existing optimization and analysis frameworks developed for dual-polarized multi-antenna communication systems cannot be directly applied to the dual-polarized RIS-enabled HMIMO due to the following two reasons. \emph{First}, unlike conventional massive MIMO with dual-polarized antennas, the radiation efficiency of antenna elements is correlated with incident angles in dual-polarized RIS-based HMIMO due to its unique free-space propagation based feeding mechanism~\cite{Xiao_helicity_2016}. Note that the feed is in the near-field of the RIS, as shown in Fig.~\ref{sysmodel}. Therefore, incident angles vary among different RIS elements. Moreover, the incident angles corresponding to different polarizations also differ, as shown in Fig.~\ref{fig_incident_angle}. As a result, the radiation efficiency varies across different RIS elements and polarizations. This discrepancy from the conventional dual-polarized massive MIMO, where antenna elements typically possess equal radiation efficiencies, indicates the inapplicability of existing frameworks. \emph{Second}, unlike the dual-polarized relay which are deployed far from the BS, the RIS-based antenna is much closer to the BS such that the BS cannot be assumed to locate in the far field of the RIS, as shown in Fig.~\ref{sysmodel}. Therefore, the frameworks developed for the dual-polarized RIS-based relay are inapplicable. 

%{We would like to point out that besides HMIMO communications, RIS has also been utilized to facilitate computational imaging. For example, the authors in~\cite{Sun_computational_2024} design a holographic RIS-aided computational imaging system, consisting of a transmitter, a rectangular target and a receiver. To accurately estimate the scattering densities of the target, a holographic RIS is deployed to focus the main energy radiated by the transmitter towards the target, where the phase pattern of the RIS is optimized and imaging performances are analyzed. However, different from the sensing task, which is }  
%holographic MIMO has also 

{We would like to point out that in addition to HMIMO communications, RIS has also been utilized to facilitate computational imaging. For example, the authors in~\cite{Sun_computational_2024} design a holographic RIS-aided computational imaging system, where a holographic RIS is deployed to focus the main energy radiated by the transmitter towards the target so as to accurately estimate the scattering densities of the target. The phase pattern of the RIS is optimized and imaging performances (i.e., normalized mean-square-error~(NMSE)) are analyzed. However, unlike existing RIS-enabled computational imaging systems, we aim to improve the ergodic capacity of RIS-enabled holographic MIMO systems. Therefore, existing optimization and analysis frameworks designed for RIS-enabled computational imaging systems cannot be directly applied.}

	To cope with the above challenges, we first model the dual-polarized RIS-based HMIMO channels. To facilitate analysis, a tight upper bound on the ergodic capacity is derived, based on which the power allocations across polarizations are optimized. {Given the optimal power allocations, we compare the multiplexing gain and system capacity between the dual-polarized and single-polarized RIS given cross-polarization discrimination~(XPD)}. Numerical results verify our analysis and demonstrate the advantage of the dual-polarized RIS. The main contributions of this paper are summarized below,
	
	\begin{itemize}
%		\item We model a single-user downlink HMIMO network, where a dual-polarized RRS is utilized at the BS as the transmit antenna for beamforming. For the purposes of analysis, an upper bound for the ergodic capacity is derived based on Jensen's inequality, whose asymptotic tightness is unveiled. Based on the the derived upper bound, the phase shfits of the RRS elements corresponding to both polarization dimensions are optimized, and the impact of the XPD is analyzed.
		\item  We model a dual-polarized RIS-aided single-user downlink HMIMO network by taking into account the angle-dependence of RIS elements and near-field spherical-wave propagation from the BS to the RIS. To facilitate analysis, an asymptotically tight upper bound on the ergodic capacity is derived, based on which optimal power allocations over different polarizations are derived in closed-form.
		
		\item To quantify the benefit brought by dual-polarizations, we compare the multiplexing gain and ergodic capacity upper bound between the dual-polarized RIS-aided system and its single-polarized counterpart given XPD. Further, a closed-form threshold is derived for the XPD, and when the XPD exceeds this threshold, the capacity upper bound on the dual-polarized system is more than twice that of the single-polarized one.
		
		\item Simulation results verify the analytical results, and the impact of the deployment of the feed and the impact of the feed gain are also numerically demonstrated.
	\end{itemize}

The rest of this paper is organized as follows. In Section~\ref{sec_system_model}, a dual-polarized RIS-aided HMIMO network is modeled. An upper bound on the ergodic capacity of such a dual-polarized system is derived in Section~\ref{sec_ergodic_capacity}, based on which the power allocations over different polarizations are optimized in Section~\ref{sec_power_polarization}. In Section~\ref{sec_cmp_single_polarized}, the performance of a dual-polarized RIS is compared against that of the single-polarized alternative, followed by simulation results presented in Section~\ref{sec_simulation}. Finally, conclusions are drawn in Section~\ref{sec_conclusion}.

{\textit{Notation}: Scalars are denoted by italic letters, and vectors and matrices are denoted by bold-face lower-case and uppercase letters, respectively. For a complex scalar $x$, $x^*$ denotes its conjugate, $|x|$ denotes its modulus, and $\angle(x)$ denotes its phase. For a vector $\bm{v}$, $\bm{v}^T$ denotes its transpose, $\|\bm{v}\|$ denotes its Euclidean norm, and $\mathrm{diag}(\bm{v})$ denotes the diagonal matrix whose diagonal element is the corresponding element in $\bm{v}$. For a matrix $\bm{F}$, $\bm{F}^{\dagger}$ denotes its conjugate transpose, and $\bm{F}(i,j)$ denotes the element in the $i$-th row and the $j$-th column. Further, $\det(\bm{F})$ represents the determinant of matrix $\bm{F}$. $\mathrm{Tr}(\bm{A})$ represents the trace of square matrix $\bm{A}$. Furthermore, $E(\cdot)$ denotes the expectation, and $\mathcal{CN}(0,1)$ denotes the zero-mean complex Gaussian distribution with unit variance.}
%	a metasurface-based relay is introduced into a downlink network to assist the communication from the BS to multiple users. A tractable upper bound of the system capacity is derived to facilitate the analysis of how many metasurface elements are required to achieve satisfactory system performance.
	
%	the downlink multi-user network aided by sis considered ,
	
%	An RRS is special kind of metasurface, where each constituting sub-wavelength elements can refract incident electromagnetic~(EM) signals and apply adjustable phase shifts. When the biased voltages applied to the diodes on the RRS elements are changed, the phase shifts induced by the RRS elements can be reconfigured, so as to direct refracted beams towards users. 
	
%	By changing the biased voltages applied to the diodes on the RRS elements, 
\begin{figure}[!t]
	\centering
	\includegraphics[width=0.35\textwidth]{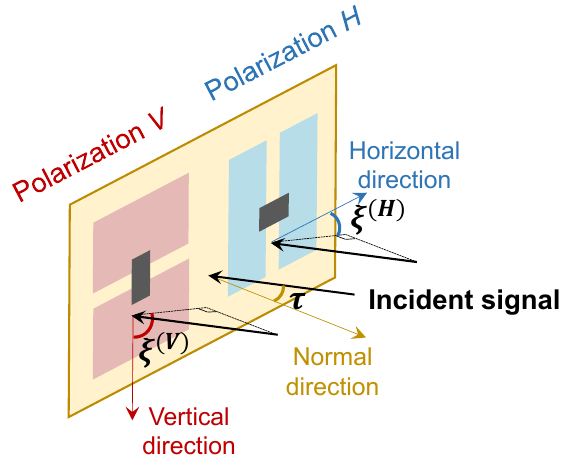}
	%				\vspace{-5mm}
	\caption{Illustration of incident angles with respect to one $V$-polarized and one $H$-polarized RIS element.}
	%		\vspace{-5mm}
	\label{fig_incident_angle}
\end{figure}
\vspace{0.5cm}
	\section{System Model}
	\label{sec_system_model}
	In this section, a dual-polarized RIS-aided downlink system is first introduced in Section~\ref{scenario_description}, where the BS employs an RIS illuminated by a feed as the transmit antenna. A model of the RIS, a propagation model from the feed to the RIS, and a propagation model from the RIS to the user equipment~(UE) are then constructed in Section~\ref{section_RRS}, Section~\ref{model_feed_2_RIS}, and Section~\ref{channel_model}, respectively.
%	\vspace{-.4cm}
	\subsection{Scenario Description}
	\label{scenario_description}
	As shown in Fig.~\ref{sysmodel}, we consider a narrow-band downlink network with one BS and one UE\footnote{The analysis methods developed for dual-polarized RIS-aided networks with one single user in this paper cannot be directly applied to multi-user networks because of the presence of inter-user interference and the unavailability of closed-form solutions for optimal RIS reflection coefficients. Due to the space limit, the extension to the multi-user case is left as future work.}. To direct the transmitted wave of the BS towards the UE, we deploy one dual-polarized RIS close to the dual-polarized feed at the BS for beamforming, where the feed is in the near-field of the RIS. {Besides, one omni-directional dual-polarized antenna is utilized at the UE}. Each dual-polarized antenna or RIS element is composed of one vertical and one horizontal polarized component that are co-located~\cite{Ozdogan_dual_polarized_2023}; these are denoted by polarizations $V$ and $H$, respectively\footnote{The analysis in this paper holds for any set of two orthogonal polarizations.}. The radiation pattern of the feed in polarization $i$ can be given by\footnote{Such radiation pattern model has been widely used in the literature~\cite{Tang_path_loss_2021,Zeng_holographic_multi_TWC_2023,stutzman_antenna_2012}, which is applicable to various antenna designs by selecting an appropriate parameter $\kappa$}~\cite{Zeng_holographic_single_2022}
	\begin{align}
		\label{pattern_feed}
		G_F^{(i)}({\bm{r}})=\left\{
		\begin{aligned}
			&\kappa(\bm{r}\cdot\bm{n})^{\frac{\kappa}{2}-1},~\bm{r}\cdot\bm{n}\ge 0,\\
			&0,~\text{otherwise}.
		\end{aligned}
		\right.
	\end{align}
	Here, $\kappa$ denotes the gain of the feed, $\bm{r}$ is a unit vector representing the transmitting direction, and $\bm{n}$ is a unit vector aligned with the main lobe direction of the feed.
	
	For ease of exposition, we introduce Cartesian coordinates, where the $yoz$ plane coincides with the RIS and the $x$-axis is vertical to the RIS, as shown in Fig.~\ref{sysmodel}. Under the introduced coordinate system, the feed's main lobe direction can be expressed as $\bm{n}=(\cos\eta,\cos\beta,\cos\gamma)$, where $\{\eta,\beta,\gamma\}$ represent the angles between $\bm{n}$ and $\{x,y,z\}$-axis, respectively. Besides, we use $r_F$, $\theta_F$, and $\phi_F$ to represent the distance between the feed and the origin, the zenith angle of the feed, and the azimuth angle of the feed, respectively. Then, the coordinate of the feed is given by $\bm{q}_F=(r_F\sin\theta_F\cos\phi_F,r_F\sin\theta_F\sin\phi_F,r_F\cos\theta_F)^{\mathrm{T}}$. Similarly, for the UE, let $r_U$ represent its distance to the center of the RIS, and $\theta_U$ and $\phi_U$ denote the zenith and azimuth angles, respectively. Therefore, the coordinate of the UE can be expressed as $\bm{q}_U=(r_U\sin\theta_U\cos\phi_U,r_U\sin\theta_U\sin\phi_U,r_U\cos\theta_U)^{\mathrm{T}}$.

%\vspace{-.8cm}
	\subsection{Dual-Polarized Reconfigurable Intelligent Surfaces}
%	\vspace{-.2cm}
	\label{section_RRS}
%	Denote the size of each RRS element by $s_{M,R}\times s_{N,R}$ and thus, the overall size of the RRS is $S_R=M_R N_Rs_{M,R}s_{N,R}$.
%\footnote{The analysis and results in this paper can be easily generalized to refractive RIS, i.e., only refract incident signals without reflections}
	The dual-polarized RIS consists of $N_R$ sub-wavelength dual-polarized reflective elements. Denote the area of each RIS element by $s_R$ and thus the overall size of the RIS is $S_R=N_Rs_R$. Besides, unlike conventional single-polarized RIS element, each dual-polarized RIS element contains two sets of diodes as shown in Fig.~\ref{sysmodel}, which can be either varactor diodes~\cite{J_varactor_2021} or positive-intrinsic-negative~(PIN)~\cite{Zeng_holographic_single_2022}. By separately changing the biased voltages applied to the two sets of diodes (i.e., \textit{ON} and \textit{OFF} states), the dual-polarized RIS element can adjust the phases of the signals in different polarizations \emph{independently}~\cite{Wu_transmitarray_2022}. 
	
	Denote the reflection amplitudes and phase shifts of the $n$-th RIS element corresponding to polarizations $V$ and $H$ by $(A^{(V)}_{n},\varphi^{(V)}_{n})$ and $(A^{(H)}_{n},\varphi^{(H)}_{n})$, respectively. Unlike conventional massive MIMO, the signal generated by the RF chain is fed into different RIS elements through free-space propagation, where the reflection amplitudes of the RIS elements are dependent on incident angles~\cite{Xiao_helicity_2016}, i.e.,
	\begin{align}
		\label{angle_dependence}
		A^{(V)}_n=A(\xi_n,\tau_n^{(V)}),\\
		A^{(H)}_n=A(\xi_n,\tau_n^{(H)}).
	\end{align} 
	Here, $(\xi_n,\tau_n^{(V)})$ and $(\xi_n,\tau_n^{(H)})$ represent the incident angles with respect to the $n$-th $V$-polarized and the $n$-th $H$-polarized RIS element, respectively, as shown in Fig.~\ref{fig_incident_angle}. Besides, functions $A(\xi,\tau)$ describes the relationship between the reflection amplitude and incident angle, which can be expressed as
	\begin{align}
		&A(\xi,\tau)\notag\\
		&\!=\!\frac{\left|\exp\left(2j\arctan(\frac{\tan(\frac{\varphi_0}{2})+\tau}{\cos\xi})\right)\!-\!\exp\left(2j\arctan(\frac{\tan(\frac{\varphi_0}{2})-\tau}{\cos\xi})\right)\right|}{2},
	\end{align}
	Here, $\varphi_0$ is the phase shift induced by the RIS element under normal incidence. Note that the azimuth component of the incident angles varies between different polarizations, i.e., $\tau_n^{(V)}\neq\tau_n^{(H)}$, as shown in Fig.~\ref{fig_incident_angle}. Therefore, the reflection amplitudes corresponding to different polarizations generally take different values, i.e., 
	\begin{align}
		\label{unequal_amplitude_polarization}
		A^{(V)}_n\neq A^{(H)}_n.
	\end{align}
	
	Then, the reflection coefficients of the $n$-th RIS element in the two polarizations can be expressed as $\Gamma^{(V)}_{n}=A^{(V)}_{n}\exp{(j\varphi^{(V)}_{n})}$ and $\Gamma^{(H)}_{n}=A^{(H)}_{n}\exp{(j\varphi^{(H)}_{n})}$, respectively. By stacking the reflection coefficients of all RIS elements, we can obtain the following matrix
	\begin{equation}
	\label{Gamma}
	\bm{\Gamma}=\left[
	\begin{matrix}
	\bm{\Gamma}^{(V)} & 0  \\
	0 & \bm{\Gamma}^{(H)}
	\end{matrix}
	\right],
	\end{equation} 
	where $\bm{\Gamma}^{(V)}=\mathrm{diag}(\Gamma^{(V)}_{1},\dots,\Gamma^{(V)}_{n},\dots,\Gamma^{(V)}_{N_R})$ and $\bm{\Gamma}^{(H)}=\mathrm{diag}(\Gamma^{(H)}_{1},\dots,\Gamma^{(H)}_{n},\dots,\Gamma^{(H)}_{N_R})$.

	\subsection{Propagation Model from the Feed to the RIS}
	\label{model_feed_2_RIS}
	Since the RIS is deployed close to the feed while environmental scatters are usually far from both the feed and the RIS, the line-of-sight~(LoS) path from the feed to the RIS is much stronger than environment scattering paths. Therefore, we only consider the LoS path here and  utilize a deterministic model to characterize the propagation effect from the feed to the RIS. Note that both the feed and the RIS elements are dual-polarized. Therefore, the propagation matrix $\bm{B}\in \mathbb{C}^{2N_R\times 2}$ from the feed to the RIS can be written as
	\begin{equation}
	\label{propagation_coeff}
	\bm{B}=\left[
	\begin{matrix}
	\bm{b}^{(VV)} & \bm{b}^{(VH)}  \\
	\bm{b}^{(HV)} & \bm{b}^{(HH)}
	\end{matrix}
	\right],
	\end{equation}
	where the $n$-th element of the vector $\bm{b}^{(ji)}\in \mathbb{C}^{N_R\times 1}$, i.e., $b^{(ji)}_n$, represents the propagation coefficient from the feed in polarization $i$ to the $n$-th RIS element in polarization $j$, for any $i,j\in\{V,H\}$. To this end, the propagation matrix $\bm{B}$ describes the relation from $V$ to $V$, $H$ to $V$, $V$ to $H$, and $H$ to $H$ polarized waves. 
	
	We would like to point out that the interactions between transmitted waves and environmental scatterers serve as a major mechanism that can change the initial polarization state of the EM waves~\cite{Degli_Esposti_polarization_2011}. However, as we have mentioned, the propagation from the feed to the RIS is LoS dominated and such interaction is weak, and thus we assume that the propagations from the feed to the RIS do not change polarizations~\cite{Ozdogan_dual_polarized_2023}, i.e.,
	\begin{align}
	\bm{b}^{(HV)}=\bm{b}^{(VH)}=\bm{0}
	\end{align}
	
	Moreover, according to~\cite{yu_dual_polarized_2022}, the co-polarized propagation coefficients can be further modeled as
%	By introducing stacking parameter $l_{FR,n}$ into a vector $\bm{q}_{F\rightarrow R}\in \mathbb{C}^{N_R}$ 
	\begin{align}
	b^{(VV)}_n=e^{j\varphi^{(VV)}}\hat{b}_n,\notag\\
	b^{(HH)}_n=e^{j\varphi^{(HH)}}\hat{b}_n,
	\end{align} 
	where $\varphi^{(VV)}$ and $\varphi^{(HH)}$ are the phase shifts in the corresponding polarizations, and $\hat{b}_n$ represents the shared component. Since the feed is in the near field of the RIS, $\hat{b}_n$ can be characterized by the Non-Uniform Spherical Wave~(NUSW) model~\cite{Zeng_holographic_multi_TWC_2023,haiquan_ELMIMO_2022},
	\begin{align}
	\label{propagation_effect}
	\hat{b}_n=\sqrt{\frac{G_{F,n} A_{F,n}}{4\pi (D_{n})^2}}\exp\left(-j\frac{2\pi}{\lambda}D_{n}\right),
	\end{align}
	where $G_{F,n}$ is the gain of the feed towards the direction of the $n$-th RIS element, $D_{n}$ represents the distance between the feed and this RIS element, $A_{F,n}$ is the projected aperture of the $n$-th RIS element towards the direction of the feed. Here, the projected aperture $A_{F,n}$ can be further expressed as $A_{F,n}=(-\bm{u}_x)^{\mathrm{T}}(\bm{q}_F-\bm{q}_{n})s_R/D_{n}$, where $\bm{q}_{n}$ and $\bm{q}_F$ represent the coordinates of the $n$-th RIS element and the feed, respectively, and $\bm{u}_x$ is the unit vector in the $x$-axis.
	\vspace{-.2cm}
	\subsection{Channel Model from the RIS to the UE}
	\label{channel_model}
	Similarly, since both the RIS elements and the UE antenna are dual-polarized, the channel matrix $\bm{H}\in \mathbb{C}^{2N_R\times 2}$ from the RIS elements to the UE consists of four components, i.e.,
	\begin{equation}
	\label{channel}
	\bm{H}=\left[
	\begin{matrix}
	\bm{h}^{(VV)} & \bm{h}^{(VH)}  \\
	\bm{h}^{(HV)} & \bm{h}^{(HH)}
	\end{matrix}
	\right],
	\end{equation}
	where, for any $i,j\in\{V,H\}$, the $n$-th element of the vector $\bm{h}^{(ji)}\in\mathbb{C}^{N_R\times 1}$, i.e., $h^{(ji)}_n$, is defined as the channel from the $n$-th RIS element in polarization $i$ to the UE antenna in polarization $j$. 
	
	We assume non-line-of-sight~(NLoS) communications between the RIS and the UE, and thus the corresponding channel can be modeled as the \emph{Rayleigh fading} channel~\cite{Ozdogan_dual_polarized_2023}, i.e.,
	\begin{align}
	\label{model_RRS_2_UE}
	h^{(ji)}_n=\sqrt{\beta^{(ji)}_n}\widetilde{h}^{(ji)}_n,
	\end{align}
	where $\beta^{(ji)}_n$ is the pathloss and $\widetilde{h}^{(ji)}_n$ represents the small scale fading with zero mean and unit variance, i.e., $\widetilde{h}^{(ji)}_n\sim \mathcal{CN}(0,1)$. 
	
	Due to  interactions between the reflected EM waves and environmental scatterers, the propagation from the RIS to the UE can change polarizations. To characterize the channel's ability to maintain radiated polarization purity between $H$ and $V$ polarized signals, cross-polarization discrimination~(XPD) serves as a common performance metric~\cite{bruno_MIMO_XPC_2013}, which is formally defined as
	\begin{align}
	\label{XPD_RRS_2_UE}
	XPD=\frac{\mathbb{E}\{|h^{(HH)}_n|^2\}}{\mathbb{E}\{|h^{(VH)}_n|^2\}}=\frac{\mathbb{E}\{|h^{(VV)}_n|^2\}}{\mathbb{E}\{|h^{(HV)}_n|^2\}}=\frac{1-l_{RU}}{l_{RU}},
	\end{align}
	for a coefficient $0\le l_{RU}\le 1$. Small values of $l_{RU}$ (i.e., high channel XPD) indicates that the received co-polarized signal at the UE is much stronger than the cross-polarized one, and thus the channels can maintaining radiated polarization purity. On the contrary, when $l_{RU}$ approaches $1$ (i.e., low channel XPD), most transmit power is converted to the orthogonal polarization. 
	
	By substituting (\ref{model_RRS_2_UE}) into (\ref{XPD_RRS_2_UE}), we have
	\begin{align}
	\label{cross_pathloss}
	\frac{\beta^{(HH)}_n}{\beta^{(VH)}_n}=\frac{\beta^{(VV)}_n}{\beta^{(HV)}_n}=\frac{1-l_{RU}}{l_{RU}}.
	\end{align}
	Before further modeling the pathloss, we would like to show a new characteristic of the channel between the RIS and the UE brought by the HMIMO. Due to the high energy and cost efficiency of RIS, it is feasible to utilize an extremely large RIS to provide significant beamforming gain. Therefore, the boundary between the near-field and far-field of the RIS, which is positively correlated with the size of RIS~\cite{YK_antenna_2008}, is comparable to the cell radius. As a result, the UE can locate in either the near-field or the far-field of the RIS. Unlike the far-field case, when the UEs are within the near field of the RIS, the variation of the received signal strength from different RIS elements should be taken into account. {To capture the characteristic of the channel in both the near-field and far-field cases, the following pathloss model is adopted~\cite{Ozdogan_dual_polarized_2023,Zeng_holographic_multi_TWC_2023}}
	\begin{align}
	\beta^{(HH)}_n&=\beta^{(VV)}_n=\beta_0d_n^{-\alpha}(1-l_{RU}),\\
	\beta^{(VH)}_n&=\beta^{(HV)}_n=\beta_0d_n^{-\alpha}l_{RU},
	\end{align}
	{where $\beta_0$ represents pathloss at unit distances}, and $d_n$ represents the distance between the $n$-th RIS element and the UE. Such model is also consistent with (\ref{cross_pathloss}).
	
	In the following, we will model the correlation structures of the channel. First, we focus on the correlation between \emph{different polarizations}. Various measurements indicate that the transmit and receive cross-polarization correlations are close to zero in NLoS scenarios~\cite{bruno_MIMO_XPC_2013}, i.e.,
	\begin{align}
	\mathbb{E}\{h^{(VV)}_n(h^{(VH)}_n)^*\}=\mathbb{E}\{h^{(HV)}_n(h^{(HH)}_n)^*\}=0,
	\end{align}
	and
	\begin{align}
	\mathbb{E}\{h^{(VV)}_n(h^{(HV)}_n)^*\}=\mathbb{E}\{h^{(HH)}_n(h^{(VH)}_n)^*\}=0,
	\end{align}
	respectively. In addition, both the co-polarization correlation and the anti-polarization correlation are also approximately zero~\cite{yu_dual_polarized_2022},
	\begin{align}
	\mathbb{E}\{h^{(VV)}_n(h^{(HH)}_n)^*\}=0,\\
	\mathbb{E}\{h^{(HV)}_n(h^{(VH)}_n)^*\}=0.
	\end{align}
	
	Since the spacing among RIS elements is less than half-wavelength, the \emph{spatial correlation} of their fading should also be taken into account. Define $\bm{R}^{(ji)}\in \mathbb{C}^{N_R\times N_R}$ as the correlation matrix of the small scale fading vector $\widetilde{\bm{h}}^{(ji)}=[\widetilde{h}_n^{(ji)}]_{n=1,\dots,N_R}$ ($i, j\in \{V,H\}$), i.e.,
	\begin{align}
	\label{def_correlation_matrix}
	\bm{R}^{(ji)}(n_1,n_2)=\mathbb{E}\{\widetilde{h}_{n_1}^{(ji)}(\widetilde{h}_{n_2}^{(ji)})^*\}.
	\end{align}
	Since the $V$- and $H$- polarized RIS elements/UE antenna are co-located, they see the same scattering environment, and thus we assume equal statistical properties~\cite{Oestges_dual_polarized_2008}, i.e.,
	\begin{align}
	\bm{R}^{(ji)}=\bm{R}, \forall i,j
	\end{align}
	where $\bm{R}$ can be further modeled by~\cite{emil_Rayleigh_2021},
	\begin{align}
	\label{correlation}
	\bm{R}(n_1,n_2)=\sinc\left(\frac{2\|\bm{q}_{n_1}-\bm{q}_{n_2}\|}{\lambda}\right).
	\end{align}
	From (\ref{correlation}), we can see that as the separation between the two RIS elements $n_1$ and $n_2$, i.e., $\|\bm{q}_{n_1}-\bm{q}_{n_2}\|$, becomes larger, the two RIS-based channels tend to be less correlated, which is consistent with practical results. 
	
	By combining (\ref{Gamma}), (\ref{propagation_coeff}), and (\ref{channel}), the received signal $\bm{y}\in\mathbb{C}^{2\times 1}$ at the UE can be written as,
	\begin{align}
	\label{rec_sig}
	\bm{y}=\bm{H}\bm{\Gamma}\bm{B}\bm{x}+\bm{n},
	\end{align}
	where $\bm{x}\in \mathbb{C}^{2\times 1}$ represents the transmitted signals of the $V$- and $H$-polarized feeds, and $\bm{n}\in \mathbb{C}^{2\times 1}$ is additive white Gaussian noise~(AWGN) at the UE with zero mean and $\sigma^2$ as variance.

	\section{Ergodic Capacity Analysis}	
	\label{sec_ergodic_capacity}
Define $\bm{G}=\bm{H}\bm{\Gamma}\bm{B}$ as the equivalent channel from the feed of the BS to the UE. Besides, assume that the channel state information~(CSI) is unknown at the transmitter. Define $\Lambda=\mathrm{diag}(\lambda^{(V)},\lambda^{(H)})$ as the power allocation matrix over different polarizations, where $\lambda^{(V)}$ and $\lambda^{(H)}$ represent transmit power allocated to polarization $V$ and $H$, respectively, satisfying $\lambda^{(H)}+\lambda^{(V)}\le 1$. Based on the received signal model in (\ref{rec_sig}), the ergodic capacity of the dual-polarized RIS-based systems is given by~\cite{Shin_Keyhole_2003,A_2005}
\begin{align}
\label{ergodic_capacity}
C_{dual}=\mathbb{E}\{\log_2\det(\bm{I}_2+\rho\bm{G}\bm{\Lambda}\bm{G}^{\dagger})\},
\end{align}
where $\mathbb{E}\{\cdot\}$ denotes expectation, $\rho$ represents transmit signal-to-noise ratio~(SNR), and $(\cdot)^{\dagger}$ denotes conjugate-transpose. Note that different from the conventional phased array, the RIS has several new degrees of freedom for design, i.e., the distance $r_F$ between the feed and the RIS, and the gain $\kappa$ of the feed. Therefore, the impacts of $r_F$ and $\kappa$ on the ergodic capacity $C_{dual}$ are discussed in the following remark.

%\footnote{Here, the RRS has the mentioned two new degrees of freedom for design compared with conventional phased arrays, rather than conventional transmit-arrays.}

%For an RRS of limited dimensions\footnote{An RRS of limited dimensions is one that only captures a part of the power radiated by the feeds.},
%\vspace{-.2cm}
\begin{remark}
	\label{remark_the_alpha_distance_on_capacity}
%	When the gain $\kappa$ of the feeds defined in (\ref{pattern_feed}) increases or the feeds move closer to the RIS, the system capacity $C_{dual}$ first increases and then decreases.
When the gain $\kappa$ of the feeds defined in (\ref{pattern_feed}) increases, the system capacity $C_{dual}$ first increases and then decreases.
\end{remark}
%\vspace{-.4cm}
{
\begin{proof}
	We first explain why the system capacity increases with the feed gain $\kappa$ when the feed gain takes a small value. This is because as the feed gain becomes larger, the beams radiated by the feeds are more concentrated, and thus more power is captured by the RIS. 
	
	However, as $\kappa$ continues to grow, the capacity will degrade. This is because most power radiated by the feeds is captured by the RIS. Besides, the received power distribution over the RIS elements is more uneven, i.e., the elements in the center of the RIS receive more power while those at the edge receive less, which leads to performance loss. 
%	We first discuss the dependence of the ergodic capacity on the feed gain $\kappa$. For each realization of the dual-polarized RIS-aided channel, it can be found that the corresponding system capacity first increases and then decreases as the feed gain $\kappa$ becomes larger, by considering the overall power captured by the RIS and the received power distribution over different RIS elements~\cite{Zeng_holographic_multi_TWC_2023}. Therefore, the ergodic system capacity in (\ref{ergodic_capacity}), which is averaged over all channel realizations, is positively correlated with the feed gain $\kappa$ when $\kappa$ is small, while is negatively correlated with $\kappa$ under a large feed gain.
%	 Similarly, we can prove the trend of the ergodic capacity with respect to the distance between the RIS and the feed. 
\end{proof}}

%(can explain: the correlation matrix cannot be expressed in Krnocknell product, and thus under limited snr, the closed-form expression of the ergodic capacity cannot be obtained.)
The analysis of the ergodic capacity in (\ref{ergodic_capacity}) is intractable. To acquire insights into how much gain can be harvested from the dual-polarization, we provide a more tractable upper bound on the ergodic system capacity in Section~\ref{subsection_upper_bound}. Then, the phase shifts of the RIS elements in different polarizations are jointly optimized to improve such upper bound in Section~\ref{subsection_phase_shift_optimization}.

\subsection{Upper Bound on the Ergodic Capacity}
\label{subsection_upper_bound}
%Before presenting the upper bound for the ergodic system capacity, we first provide the expressions of the entries of the equivalent channel matrix $\bm{G}$, which will utilized for the derivation of the upper bound. 
%\begin{lemma}
%	By substituting (\ref{Gamma}), (\ref{propagation_coeff}), and (\ref{channel}) into the definition of the equivalent channel matrix $\bm{G}=\bm{H}\bm{\Gamma}\bm{B}$, its entries can be written as
%	\begin{align}
%	\label{G_11}
%	G_{1,1}=\bm{h}^{(VV)}\bm{\Gamma}^{(V)}\bm{b}^{(VV)},\\
%	\label{G_12}
%	G_{1,2}=\bm{h}^{(VH)}\bm{\Gamma}^{(H)}\bm{b}^{(HH)},\\
%	\label{G_21}
%	G_{2,1}=\bm{h}^{(HV)}\bm{\Gamma}^{(V)}\bm{b}^{(VV)},\\
%	\label{G_22}
%	G_{2,2}=\bm{h}^{(HH)}\bm{\Gamma}^{(H)}\bm{b}^{(HH)}.
%	\end{align}
%\end{lemma}

By substituting (\ref{Gamma}), (\ref{propagation_coeff}), and (\ref{channel}) into the definition of the equivalent channel matrix $\bm{G}$, i.e., $\bm{G}=\bm{H}\bm{\Gamma}\bm{B}$, the entries of $\bm{G}$ can be rewritten as
\begin{align}
\label{G_11}
G_{1,1}=\bm{h}^{(VV)}\bm{\Gamma}^{(V)}\bm{b}^{(VV)},\\
\label{G_12}
G_{1,2}=\bm{h}^{(VH)}\bm{\Gamma}^{(H)}\bm{b}^{(HH)},\\
\label{G_21}
G_{2,1}=\bm{h}^{(HV)}\bm{\Gamma}^{(V)}\bm{b}^{(VV)},\\
\label{G_22}
G_{2,2}=\bm{h}^{(HH)}\bm{\Gamma}^{(H)}\bm{b}^{(HH)},
\end{align}
based on which an upper bound on the ergodic capacity can be derived, as shown in the following theorem.
\begin{theorem}
	\label{theorem_upper_bound}
	The ergodic system capacity in (\ref{ergodic_capacity}) can be upper bounded by
\begin{align}
C_{dual}\le \log_2\left(\mathbb{E}\{\det(\bm{I}_2+\rho\bm{G}\bm{\Lambda}\bm{G}^{\dagger})\}\right)\triangleq C_{dual}^{ub}.
\end{align}
To facilitate analysis, we rewrite the upper bound as in (\ref{C_R_ub_1}) shown at the bottom of the next page.
\begin{figure*}[!hb]
	%	{\noindent}
	\rule[-12pt]{17.5cm}{0.05em}
	\begin{equation}
	\setlength{\abovedisplayskip}{3pt}
	\setlength{\belowdisplayskip}{-3pt}
%	\begin{aligned}
%	\label{C_R_ub_1}
%	C_{dual}&\le \log_2\left(\mathbb{E}\{\det(\bm{I}_2+\rho\bm{G}\bm{\Lambda}\bm{G}^{\dagger})\}\right)\\
%%	\label{C_R_ub}
%	&=\log_2\bigg(1+\rho\lambda^{(V)}\left(\mathbb{E}\{|G_{1,1}|^2+|G_{2,1}|^2\}\right)+\rho\lambda^{(H)}\left(\mathbb{E}\{|G_{1,2}|^2+|G_{2,2}|^2\}\right)\\
%	&\quad\quad\quad+\!\rho^2\lambda^{(V)}\lambda^{(H)}\left(\mathbb{E}\{|G_{1,1}|^2\}\mathbb{E}\{|G_{2,2}|^2\}\!+\!\mathbb{E}\{|G_{1,2}|^2\}\mathbb{E}\{|G_{2,1}|^2\}\right)\!\bigg)\\
%	&\triangleq C_{dual}^{ub},
%	\end{aligned}
\begin{aligned}
\label{C_R_ub_1}
C_{dual}^{ub}&=\log_2\bigg(1+\rho\lambda^{(V)}\left(\mathbb{E}\{|G_{1,1}|^2+|G_{2,1}|^2\}\right)+\rho\lambda^{(H)}\left(\mathbb{E}\{|G_{1,2}|^2+|G_{2,2}|^2\}\right)\\
&\quad\quad\quad+\!\rho^2\lambda^{(V)}\lambda^{(H)}\left(\mathbb{E}\{|G_{1,1}|^2\}\mathbb{E}\{|G_{2,2}|^2\}\!+\!\mathbb{E}\{|G_{1,2}|^2\}\mathbb{E}\{|G_{2,1}|^2\}\right)\!\bigg),
\end{aligned}
	\end{equation}
\end{figure*}
\end{theorem}
\begin{proof}
See Appendix~\ref{appendix_upper_bound}.
\end{proof}
\begin{remark}
	\label{remark_asymptotically_tight}
	The derived upper bound $C_{dual}^{ub}$ is asymptotically tight as the SNR $\rho$ tends to zero.
\end{remark}
\begin{proof}
	See Appendix~\ref{appendix_asymptotically_tight}.
\end{proof}

%\begin{align}
%\max_{\bm{V}_D,\bm{\Phi},\bm{\Phi}}\sum_m\log_2(1+\frac{|\bm{h}_m^H\bm{\Phi}BV_D^{(m)}|^2}{\sum_{j\neq m}|\bm{h}_j^H\bm{\Phi}BV_D^{(j)}|^2+\sigma^2})
%\end{align}
%Based on (\ref{})

%\begin{remark}
%	\label{remark_tighter_parameter}
%%	When the RRS can capture most power radiated by the feed (e.g., when the RRS is sufficiently large), the upper bound $C_R^{ub}$ will be tighter when the feed moves away from the RRS or the gain of the feed becomes larger.
%When the feed moves away from the RRS or the gain of the feed becomes larger, the gap between the system capacity and its upper bound becomes larger first, and then the upper bound becomes tighter.
%\end{remark}
%\begin{proof}
%	See Appendix~\ref{appendix_tighter_parameter}
%\end{proof}
\subsection{Optimization of RIS Phase Shifts}
\label{subsection_phase_shift_optimization}
According to (\ref{C_R_ub_1}), we can find that the upper bound $C_{dual}^{ub}$ depends on the reflection phase shifts of the RIS. To improve the system capacity, in the following theorem, the phase shifts $\varphi_n^{(V)}$ and $\varphi_n^{(H)}$ in the two polarizations are optimized.
\begin{theorem}
	\label{the_opt_phase}
	To improve the capacity upper bound $C_{dual}^{ub}$, the phase shifts of the RIS elements should be tuned such that the phases of the reflected signals of different RIS elements are aligned, i.e., 
%	When the transmission phase shifts $\varphi_n^{(V)}$ and $\varphi_n^{(H)}$ of the $n$-th RRS element corresponding to polarizations $V$ and $H$ exactly compensate the phase shifts induced by the propagations from the feed to the RRS, i.e.,
	\begin{align}
	\label{opt_phase}
	\varphi_n^{(V)}=\varphi_n^{(H)}=-\angle \hat{b}_n=\frac{2\pi}{\lambda}D_n,
	\end{align}
	where $\hat{b}_n$ represents the propagation effect from the feed to the RIS defined in (\ref{propagation_effect}). 
	
	Given the phase shifts in (\ref{opt_phase}), we can rewrite the capacity upper bound in the form of equation (\ref{C_R_ub_max}), shown at the bottom of this page,
	\begin{figure*}[!hb]
		%	{\noindent}
		\rule[-12pt]{17.5cm}{0.05em}
		\begin{equation}
		\setlength{\abovedisplayskip}{3pt}
		\setlength{\belowdisplayskip}{0pt}
		\label{C_R_ub_max}
		C_{dual}^{ub}=\log_2\left(1+\rho \left(\lambda^{(H)}O^{(H)}+\lambda^{(V)}O^{(V)}\right)+\rho^2 \lambda^{(H)}\lambda^{(V)}O^{(H)}O^{(V)}\left(l_{RU}^2+(1-l_{RU})^2\right)\right),
		\end{equation}
	\end{figure*}
%	\begin{align}
%		\label{C_R_ub_max}
%	C_{dual}^{ub}=\log_2(&1+\frac{\rho}{2}O((\Gamma^{(H)})^2+(\Gamma^{(V)})^2)\notag\\
%	&+\frac{\rho^2}{4}O^2(\Gamma^{(H)})^2(\Gamma^{(V)})^2(l_{RU}^2+(1-l_{RU})^2)),
%	\end{align}
	where $l_{RU}$ characterizes the XPD of the channel from the RIS to the UE. Besides, the quantities of $O^{(V)}$ and $O^{(H)}$ in (\ref{C_R_ub_max}) are given by
	\begin{align}
		\label{O_V}
	O^{(V)}=\sum_{n_1}\sum_{n_2}A_{n_1}^{(V)}A_{n_2}^{(V)}|\hat{b}_{n_1}||\hat{b}_{n_2}|\bm{R}(n_1,n_2)\beta_0\sqrt{d_{n_1}^{-\alpha}d_{n_2}^{-\alpha}},\\
	\label{O_H}
	O^{(H)}=\sum_{n_1}\sum_{n_2}A_{n_1}^{(H)}A_{n_2}^{(H)}|\hat{b}_{n_1}||\hat{b}_{n_2}|\bm{R}(n_1,n_2)\beta_0\sqrt{d_{n_1}^{-\alpha}d_{n_2}^{-\alpha}}.
	\end{align}
	where $A_{n}^{(V)}$ and $A_{n}^{(H)}$ are the reflection amplitudes of the RIS elements under different polarizations.
\end{theorem} 
\begin{proof}
	See Appendix~\ref{app_opt_phase}
\end{proof}
%\vspace{.1cm}
Based on Theorem~\ref{the_opt_phase}, we can derive the following two remarks.

\begin{remark}
	The optimal reflection phase shifts depend only on the setup of the BS, i.e., the relative deployment location of the feed with respect to the RIS. Therefore, the configuration of the reflection phase shifts does not require any CSI between the RIS and the UE, which indicates that the proposed phase shift configuration scheme matches the considered case where the CSI is unknown to the BS.
\end{remark}
%\vspace{.05cm}
\begin{definition}
	Similar to the discussions in~\cite{yu_dual_polarized_2022}, we define \emph{phase adjustment between two polarizations} as the difference of the optimal phase shifts between the two polarizations.
\end{definition} 
%\vspace{.05cm}
\begin{remark}
	\label{remark_phase_adjustment}
	According to (\ref{opt_phase}), phase adjustment between two polarizations are not required to maximize the system capacity in an RIS-aided dual-polarized system.
\end{remark}
\vspace{.1cm}
%\begin{align}
%\label{E_G_11_module_square}
%\mathbb{E}(|G_{1,1}|^2)=\sum_{n_1}\sum_{n_2}b_{n_1}^{(VV)}(b_{n_2}^{(VV)})^*\Gamma_{n_1}^{(V)}(\Gamma_{n_2}^{(V)})^*\bm{R}(n_1,n_2)\sqrt{\beta_{n_1}^{(VV)}\beta_{n_2}^{(VV)}},\\	
%\label{E_G_12_module_square}
%\mathbb{E}(|G_{1,2}|^2)=\sum_{n_1}\sum_{n_2}b_{n_1}^{(HH)}(b_{n_2}^{(HH)})^*\Gamma_{n_1}^{(H)}(\Gamma_{n_2}^{(H)})^*\bm{R}(n_1,n_2)\sqrt{\beta_{n_1}^{(VH)}\beta_{n_2}^{(VH)}},\\	
%\label{E_G_21_module_square}
%\mathbb{E}(|G_{2,1}|^2)=\sum_{n_1}\sum_{n_2}b_{n_1}^{(VV)}(b_{n_2}^{(VV)})^*\Gamma_{n_1}^{(V)}(\Gamma_{n_2}^{(V)})^*\bm{R}(n_1,n_2)\sqrt{\beta_{n_1}^{(HV)}\beta_{n_2}^{(HV)}}\\	
%\label{E_G_22_module_square}
%\mathbb{E}(|G_{2,2}|^2)=\sum_{n_1}\sum_{n_2}b_{n_1}^{(HH)}(b_{n_2}^{(HH)})^*\Gamma_{n_1}^{(H)}(\Gamma_{n_2}^{(H)})^*\bm{R}(n_1,n_2)\sqrt{\beta_{n_1}^{(HH)}\beta_{n_2}^{(HH)}}.
%\end{align}

%According to the definition of the equivalent channel $\bm{G}=\bm{H}\bm{\Gamma}\bm{B}$, its entry $G_{1,1}$ can be rewritten as $G_{1,1}=\bm{h}^{(VV)}\bm{\Gamma}^{(V)}\bm{b}^{(VV)}$, based on which we have
%
% Similarly, we can prove (\ref{E_G_12_module_square})-(\ref{E_G_22_module_square}), and the proof is neglected here due to the space limit. 
\vspace{-0.2cm}
\section{Power Allocation across Polarizations}
\label{sec_power_polarization}
The introduction of dual-polarization provides a new degree of freedom for system design, i.e., power allocations $\lambda^{(V)}$ and $\lambda^{(H)}$ over different polarization directions. In this section, we aim to optimize the power allocations $\lambda^{(V)}$ and $\lambda^{(H)}$ so as to maximize the capacity upper bound $C_{dual}^{ub}$ in (\ref{C_R_ub_max}). Formally, the power allocation problem can be written as
\begin{subequations}\label{power_allocation_polarization}
	\begin{align}
	&\max_{\lambda^{(V)},\lambda^{(H)}} C_{dual}^{ub}\\
	\label{cons_1}
	s.t.&~\lambda^{(V)}+\lambda^{(H)}\le 1,\\
	&~\lambda^{(V)},\lambda^{(H)}\ge 0.
	\end{align}
\end{subequations}
According to (\ref{C_R_ub_max}), we can find that $C_{dual}^{ub}$ is increasing in both $\lambda^{(V)}$ and $\lambda^{(H)}$. Therefore, when the power allocations are maximized, the total transmit power achieves the maximum allowed transmit power, i.e., $\lambda^{(V)}+\lambda^{(H)}= 1$. Then, by substituting $\lambda^{(H)}=1-\lambda^{(V)}$ into (\ref{power_allocation_polarization}), the power allocation problem can be reformulated as
\begin{subequations}\label{power_allocation_polarization_v2}
	\begin{align}
	&\max_{\lambda^{(V)}} C_{dual}^{ub}\\
	\label{cons_1_v2}
	s.t.&~0\le \lambda^{(V)}\le 1,
	\end{align}
\end{subequations}
based on which we can acquire the following theorem.
\begin{theorem}
	The optimal power allocations for problem (\ref{power_allocation_polarization}) are given by
	\begin{align}
	\label{optimal_power_allocation_V}
	(\lambda^{(V)})^*=\left\{
	\begin{aligned}
	&1,~\lambda_0>1,\\
	&\lambda_0,~\lambda_0\in (0,1),\\
	&0,~\lambda_0<0.
	\end{aligned}
	\right.
	\end{align}
	\begin{align}
%	\label{optimal_power_allocation_V}
%	(\lambda^{(V)})^*=\frac{1}{2}+\frac{1}{2\rho O(l_{RU}^2+(1-l_{RU})^2)}\frac{(\Gamma^{(V)})^2-(\Gamma^{(H)})^2}{(\Gamma^{(V)})^2(\Gamma^{(H)})^2},\\
	\label{optimal_power_allocation_H}
%	(\lambda^{(H)})^*=\frac{1}{2}+\frac{1}{2\rho O(l_{RU}^2+(1-l_{RU})^2)}\frac{(\Gamma^{(H)})^2-(\Gamma^{(V)})^2}{(\Gamma^{(V)})^2(\Gamma^{(H)})^2},
(\lambda^{(H)})^*=1-(\lambda^{(V)})^*,
	\end{align}
	where $\lambda_0$ is given by
	\begin{align}
		\label{optimal_power_allocation_lambda_0}
	\lambda_0=\frac{1}{2}+\frac{1}{2\rho \left(l_{RU}^2+(1-l_{RU})^2\right)}\frac{O^{(V)}-O^{(H)}}{O^{(V)}O^{(H)}}.
	\end{align}
	Moreover, the maximized $C_{dual}^{ub}$ is given by (\ref{C_R_ub_max_max}) at the bottom of the next page.
	\begin{figure*}[!hb]
		%	{\noindent}
		\rule[-12pt]{17.5cm}{0.05em}
		\begin{equation}
		\setlength{\abovedisplayskip}{3pt}
		\setlength{\belowdisplayskip}{3pt}
		\label{C_R_ub_max_max}
		C_{dual}^{ub}=\log_2\left(1+\rho \left((\lambda^{(H)})^*O^{(H)}+(\lambda^{(V)})^*O^{(V)}\right)+\rho^2 (\lambda^{(H)})^*(\lambda^{(V)})^*O^{(H)}O^{(V)}\left(l_{RU}^2+(1-l_{RU})^2\right)\right).
		\end{equation}
	\end{figure*}
\end{theorem}
\begin{proof}
	Note that problem (\ref{power_allocation_polarization_v2}) is a univariate optimization problem with respect to $\lambda^{(V)}$. Therefore, by studying the derivative of $C_{dual}^{ub}$ with respect to $\lambda^{(V)}$, we can acquire the optimal power allocation in (\ref{optimal_power_allocation_V}) and (\ref{optimal_power_allocation_H}). Due to space limitations, this straightforward derivation is omitted here.
\end{proof}
Based on (\ref{optimal_power_allocation_V}) and (\ref{optimal_power_allocation_H}), we have the following remarks concerning the optimal power allocations across polarizations.
\begin{remark}
	In contrast to conventional massive MIMO systems utilizing dual-polarized antennas, where transmit power is uniformly distributed across polarizations, the optimal power allocation for the dual-polarized RIS-aided system varies between polarizations. This variation arises from the differences in reflection amplitudes associated with different polarizations due to the angle-dependence of the RIS elements, as referenced in (\ref{angle_dependence})-(\ref{unequal_amplitude_polarization}). Moreover, based on (\ref{O_V}) and (\ref{O_H}), the BS should allocate more power to the polarization direction corresponding to larger average reflection amplitudes of the RIS elements.
%	the BS should allocate more power to polarization $V$
%	when the average reflection amplitude in polarization $V$ over different RIS elements is larger than that in polarization $H$, we have $O^{(V)}>O^{(H)}$ according to (\ref{O_H}) and (\ref{O_V}). Therefore, based on (\ref{optimal_power_allocation_V})-(\ref{optimal_power_allocation_lambda_0}), the BS should allocate more power to polarization $V$, i.e., $(\lambda^{(V)})^*>(\lambda^{(H)})^*$. Otherwise, we have $(\lambda^{(V)})^*\le(\lambda^{(H)})^*$.
\end{remark}
\begin{remark}
	\label{remark_rho_N}
%	When the SNR $\rho$ or the number of the RRS elements is infinitely large, the transmit power should be allocated equally between different polarizations.
As the SNR $\rho$ approaches infinity, the transmit power tends to be evenly distributed between different polarizations.
\end{remark}
%formulate最优化问题
%
%分析最优使得取到等号，reformulate最优化问题
%
%给个theorem，说一下最优分配，以及最大容量上限。proof一笔带过就好。
%
%给三个remark，第一个讲gamma的影响，第二个讲rho增加时的变化，第三个讲单元数目增加时的变化。

%\section{Effect of Cross-Polarization Discrimination}
%\label{sec_XPD}
%In this section, we will investigate the influence of XPD, a new parameter for the dual-polarized RIS-based channel, on the system capacity.
%The following theorem discusses the influence of XPD associated with the channel from the RIS to the UE on the system capacity.
Given the optimal power allocations, the influence of XPD, a new parameter for the dual-polarized RIS-based channel, can be analyzed, as indicated in the following theorem.
\begin{theorem}
	\label{theorem_XPD}
	As the cross-polarization coefficient $l_{RU}$ given in (\ref{XPD_RRS_2_UE}) increases from $0$ to $1$ (i.e., XPD changes from $+\infty$ to $0$), the system capacity $C_{dual}^{ub}$ first degrades and then becomes larger. 
	
	Moreover, $C_{dual}^{ub}$ is maximized when $l_{RU}=0$ or $l_{RU}=1$ (i.e., XPD$=+\infty$ or XPD$=0$) while it is minimized at $l_{RU}=\frac{1}{2}$ (i.e., XPD$=1$).
\end{theorem}
%\vspace{.1cm}
\begin{proof}
	A proof of these results can be found in Appendix~\ref{appendix_XPD}. Here, we give a physical explanation in order to provide some insight.
	
	Note that when $l_{RU}=0$, the polarization of the transmitted signal will not change when it propagates from the feed to the UE. Alternatively, $l_{RU}=1$ corresponds to the case where the transmitted signal by the feed will be completely converted to the orthogonal polarization when it is received by the UE. Therefore, in both cases, there is no cross-polarization interference at the UE, and the dual-polarized channel can be decomposed to two parallel and independent subchannels. On the other hand, when $0<l_{RU}<1$, a part of the transmit power is preserved in the original polarization while the rest is coupled to the orthogonal polarization, and thus cross-polarization always exists in the received signals, leading to degraded system capacity. Therefore, the system capacity can be maximized at $l_{RU}=0$ and $l_{RU}=1$.
	
	Define $x$ as a transmitted signal in the $V$- or $H$- polarization. Also, use $y_{C}$ and $y_{X}$ to represent the induced co-polarized and cross-polarized received signals, respectively. When $l_{RU}\neq \frac{1}{2}$, the power of $y_{C}$ is unequal to that of $y_{X}$ according to (\ref{XPD_RRS_2_UE}). To this end, we can always regard the component with higher power, i.e., $\arg\max\{\|y_{C}\|^2,\|y_{X}\|^2\}$, as the desired signal, and the other as cross-polarization interference, which indicates that the desired signal is stronger than the interference. However, when $l_{RU}=\frac{1}{2}$, the received signals $y_{C}$ and $y_{X}$ have the same strength, which indicates that the desired signal associated with $x$ has the same strength as the interference caused by $x$. Therefore, when $l_{RU}=\frac{1}{2}$, the system capacity achieves its minimum value.
\end{proof}

%Based on the system capacity in (\ref{C_R_ub_max_max}), the following remark on the optimal orientation of the feed can be obtained to further improve the system capacity. 
%\begin{remark}
%	\label{remark_orientation}
%	The optimal orientation of the feed $\bm{n}$ is approximately aligned with the direction from the feed to the center of the RIS, i.e., the feed should be directed towards the center of the RIS.
%\end{remark}
%\begin{proof}
%	See Appendix~\ref{app_orientation}.
%\end{proof}

\vspace{0.2cm}

\section{Comparison with Single-Polarized RIS-Aided System}
\label{sec_cmp_single_polarized}
Before comparing the performance of the single-polarized RIS-aided system against that of the dual-polarized system, we first derive the ergodic capacity of the single-polarized system and an upper bound on it, which is then maximized by optimizing the phase shift configuration of the single-polarized RIS, as shown in the following lemma. Without loss of generality, we assume that the feed, the RIS and the UE antenna in the single-polarized system are all $V$-polarized while other system settings are the same as those of the considered dual-polarized system. 
%\footnote{We would like to point out that unlike Remark~\ref{remark_dual_equivalent_single} where the dual-polarized RIS-aided system is compared against two single-polarized system, each with half of the transmit power budget of the dual-polarized system, in this part, we only consider one single-polarized system with the same transmit power budget, so as to demonstrate the gain that can be harvested from the dual-polarization}
% of each single-polarized system is half of that of the considered dual-polarized system
\begin{lemma}
	\label{lemma_single_polarization}
	The ergodic capacity of the single-polarized RIS-aided systems is given by
	\begin{align}
	\label{ergodic_capacity_single}
	C_{single}=\mathbb{E}\{\log_2(1+\rho|G_{1,1}|^2)\},
	\end{align}
	where the definition of $G_{1,1}$ can be found in (\ref{G_11}). 
\end{lemma}
\begin{proof}
	See Appendix~\ref{appendix_single_polarization}.
\end{proof}
Based on (\ref{ergodic_capacity_single}) and (\ref{ergodic_capacity}), we have the following theorem on the asymptotic performance of the single-polarized and the dual-polarized RIS-aided system. 
\begin{theorem}
	\label{theorem_multiplexing_gain_cmp}
	The multiplexing gain of the dual-polarized system doubles that of the single-polarized system, i.e.,
	\begin{align}
	\label{multiplexing_gain_dual}
	\lim_{\rho\rightarrow +\infty}\frac{C_{dual}}{\log_2\rho}=2,\\
	\label{multiplexing_gain_single}
	\lim_{\rho\rightarrow +\infty}\frac{C_{single}}{\log_2\rho}=1.
	\end{align}
\end{theorem}
\begin{proof}
	See Appendix~\ref{appendix_multiplexing_gain_cmp}.
\end{proof}
\begin{lemma}
	Similar to the dual-polarized RIS-aided system, by applying Jensen's inequality, an upper bound on the ergodic capacity $C_{single}$ of the single-polarized system can be derived, i.e.,
	\begin{align}
	\label{ergodic_capacity_single_upper_bound}
	C_{single}&\le \log_2\mathbb{E}\{1+\rho|G_{1,1}|^2\}\notag\\
	&=\log_2\left(1+\rho\mathbb{E}\{|G_{1,1}|^2\}\right)\notag\\
	&\triangleq C_{single}^{ub}.
	\end{align}
	Similar to (\ref{opt_phase}), by setting the phase shifts of the RIS elements as
	\begin{align}
	\label{opt_phase_single}
	\varphi_n^{(V)}=-\angle \hat{b}_n=\frac{2\pi}{\lambda}D_n,
	\end{align}
	the capacity upper bound can be maximized as
	\begin{align}
	\label{ergodic_capacity_single_upper_bound_opt_phase}
%	C_{single}^{ub}&=\log_2(1+\rho\mathbb{E}(|G_{1,1}|^2))\notag\\
%	&=\log_2(1+\rho(\Gamma^{(V)})^2(1-l_{RU})O)
C_{single}^{ub}=\log_2\left(1+\rho(1-l_{RU})O^{(V)}\right).
	\end{align}
\end{lemma}
In the following, we will compare the capacity upper bound on the single-polarized system in (\ref{ergodic_capacity_single_upper_bound_opt_phase}) against that of the dual-polarized system in (\ref{C_R_ub_max}) given the XPD. To simplify the discussion, we first define several channel conditions based on the relationship between the XPD associated with the channel from the RIS to the UE and the polarization of the single-polarized RIS-aided system.
%\begin{definition}
%For a single-polarized RRS-aided system, we say the channel from the RRS to the UE \textbf{matches} the polarization of the system if one of the following conditions is satisfied:
% \begin{itemize}
% 	\item The polarization of the RRS and that of the UE  antenna are the same, and the cross-polarization coefficient $l_{RU}\rightarrow 0$.
%% 	the cross-polarization component induced by the channel is weak, i.e., 
% 	\item The polarization of the RRS and that of the UE  antenna are orthogonal, and $l_{RU}\rightarrow 1$.
% \end{itemize}
%Under the two cases, if $l_{RU}$ is exactly equal to $0$ and $1$, respectively, we say the channel \textbf{perfectly matches} the polarization of the system. On the contrary, the \textbf{mismatch} between the channel and the polarization of the system can be defined as,
%\begin{itemize}
%	\item The polarization of the RRS and that of the UE  antenna are the same, and $l_{RU}\rightarrow 1$.
%	\item The polarization of the RRS and that of the UE  antenna are orthogonal, and $l_{RU}\rightarrow 0$.
%\end{itemize}
%Correspondingly, we can define \textbf{absolute mismatch}, i.e., $l_{RU}=1$ and $l_{RU}=0$ for the two cases, respectively.
%%Besides, we say the channel \textbf{perfectly matches} the polarization of the system if $l_{RU}$ is exactly equal to $0$ and $1$, respectively.
%\end{definition}
\begin{definition}
	For the single-polarized RIS-aided system, we say the channel from the RIS to the UE \textbf{matches} the polarization of the system if when the EM wave reflected by the RIS propagates to the vicinity of the single-polarized UE antenna, the polarization of the EM wave is consistent with that of the UE antenna with weak cross-polarization components. Since both the RIS and the UE antenna in the considered RIS-aided system are $V$-polarized, this corresponds to the case where $l_{RU}\rightarrow 0$ (i.e., XPD$\rightarrow \infty$). Especially, if the polarization of the EM wave is exactly the same as that of the UE antenna, the channel \textbf{perfectly matches} the polarization of the RIS-aided system, i.e., $l_{RU}=0$.
	
	On the contrary, if the polarization of the received EM wave is approximately orthogonal to that of the UE antenna, and thus the received signal is weak, we say the channel from the RIS to the UE \textbf{mismatches} the RIS-aided system, which corresponds to $l_{RU} \rightarrow 1$ (i.e., XPD$\rightarrow 0$). Further, we refer to the case where $l_{RU}=1$ as the RIS-based channel \textbf{absolutely mismatches} the single-polarized RIS-aided system.
\end{definition}
The comparison between the single-polarized and dual-polarized systems are performed under these channel conditions, respectively, as shown in the following remarks.
\begin{remark}
	\label{remark_single_polarized_2}
	%	When the capacity upper bound $C_{single}^{ub}$ of the single-polarized system is minimized, i.e., when the channel absolutely mismatches the polarization of the system, the average received SNR is $0$. Differently, the minimum received SNR of the dual-polarized system is at least half of the maximum received SNR, which demonstrates that the dual-polarized system is more robust to the wireless propagation environment than the single-polarized couterpart.
	When the channel absolutely mismatches the polarization of the single-polarized system, the average received SNR of the single-polarized system is minimized as $0$. Differently, the minimum received SNR of the dual-polarized system is non-zero, which demonstrates that the dual-polarized system is more robust to the wireless propagation environment than its single-polarized counterpart.
\end{remark}
\begin{remark}
	\label{remark_single_polarized_1}
	The capacity upper bound on the single-polarized system changes monotonically with the XPD. Besides, when the RIS-aided channel perfectly matches the polarization of the single-polarized system, the capacity upper bound $C_{single}^{ub}$ is maximized, while the upper bound is minimized when the channel absolutely mismatches the polarization of the system, which is different from the dual-polarized system indicated in Theorem~\ref{theorem_XPD}.
\end{remark}
\begin{remark}
	\label{remark_single_polarized_3}
%	When the channel matches the polarization of the single-polarized system, the introduction of dual-polarization cannot double of the capacity upper bound $C_{single}^{ub}$ of the single-polarized system. 
	When the channel mismatches the polarization of the RIS-aided system, the system capacity achieved by the dual-polarized system is more than twice that of the single-polarized one. Moreover, a closed-form threshold $l_{RU}^{(th)}$ for the polarization parameter $l_{RU}$ can be derived, i.e., when $l_{RU}>l_{RU}^{(th)}$, we have $C_{dual}^{ub}>2C_{single}^{ub}$. An expression for $l_{RU}^{(th)}$ can be given by
	\begin{align}
		\label{threshold_XPD}
		l_{RU}^{(th)}=\frac{-b+\sqrt{b^2-4ac}}{2a},
	\end{align}
	where
	\begin{align}
		a=\rho^2O^{(V)}\left(\frac{1}{2}O^{(H)}-O^{(V)}\right),
	\end{align}
	\begin{align}
		b=\rho^2O^{(V)}\left(2O^{(V)}-\frac{1}{2}O^{(H)}\right)+2\rho O^{(V)},
	\end{align}
	and
	\begin{align}
		c=\rho^2O^{(V)}\left(\frac{1}{4}O^{(H)}-O^{(V)}\right)+\rho\left(\frac{1}{2}O^{(H)}-\frac{3}{2}O^{(V)}\right),
	\end{align}
\end{remark}	
\begin{proof}
	See Appendix~\ref{app_single_polarized_3}.
\end{proof}
It is easy to generalize Remarks~\ref{remark_single_polarized_2}-\ref{remark_single_polarized_3} to other single-polarized RIS-aided systems, namely an RIS-aided system with a $V$-polarized feed and $H$-polarized UE antenna, an RIS-aided system with an $H$-polarized feed and $H$-polarized UE antenna, and an RIS-aided system with an $H$-polarized feed and $V$-polarized UE antenna.
%\vspace{0.3cm}
\section{Simulation Results}
\label{sec_simulation}
To validate the theoretical analysis, we numerically evaluate the performance of the dual-polarized RIS-aided communication system via simulations. The simulation parameters are set up based on the existing works~\cite{Zhang_IOS_2022,Zeng_holographic_multi_TWC_2023,yu_dual_polarized_2022}. Specifically, until further mentioned, the coordinates of the feed are set as $(r_F,\theta_F,\phi_F)=(0.05~m, \frac{\pi}{2}, \pi)$, with the orientation of the feed vertical to the RIS, i.e., $(\eta,\beta,\gamma)=(0, \frac{\pi}{2}, \frac{\pi}{2})$. For comparison, a single-polarized RIS-aided system is also considered and evaluated, where the feed, the RIS elements, and the UE antenna are all $V$-polarized\footnote{Throughout the simulation, the parameter $\rho$ indicates the ratio between the transmit power of the BS and noise variance, i.e., $\rho=P/\sigma^2$.}. The other simulation parameters are summarized in Table~\ref{sim_par}. 
%Specifically, until it is further mentioned, the coordinate of the feed is set as $(r_F,\theta_F,\phi_F)=()$

\renewcommand\arraystretch{1.4}
\begin{table}[!tpb]
	\footnotesize
	\centering
	\caption{\normalsize{Simulation parameters}}
	\label{sim_par}
	\centering
	\begin{tabular}{|m{150pt}|m{60pt}|}
		\hline	
		\textbf{Parameters} & \textbf{Values} \\
		\hline\hline
		Carrier frequency $f_c$ & $26$~GHz\\
		\hline
		Wavelength $\lambda$ corresponding to the carrier frequency & $1.15$~cm\\
%		\hline
%		Transmit power $P$ & $43$~dBm\\
		\hline
		Variance $\sigma^2$ of the received AWGN & $-96$~dBm\\
		\hline
		Pathloss parameter $\beta_0$ & $-49.7$~dB\\
		\hline
		Pathloss exponent $\alpha$ & $4$ \\
		\hline
		Area $s_{R}$ of each RIS element & $\frac{\lambda}{3}\times \frac{\lambda}{3}$\\		
%		\hline
%		Reflection amplitude $\Gamma^{(V)}$ and $\Gamma^{(H)}$ of the RIS element corresponding to polarization $V$ and $H$, respectively & $0.8$\\
		\hline
		Phase shifts $(\varphi^{(VV)},\varphi^{(HH)})$ induced by the propagation from the feed to the RIS & $(\frac{\pi}{2},\frac{\pi}{4})$\\
		\hline
		Coordinate $(r_U,\theta_U,\phi_U)$ of the UE & $(50~m, \frac{\pi}{3}, 0)$\\
		\hline
	\end{tabular}
%	\vspace{-.3cm}
\end{table}
%where the average received SNR is $0$. Unlike the single-polarized system, the minimum received SNR for the dual-polarized system is at least 

%\begin{figure}[!tpb]
%	\centering
%	\vspace{0pt}
%	\includegraphics[width=0.45\textwidth]{capacity_vs_feed_distance_v3.eps}
%	%	\vspace{-0.25cm}
%	\caption{System capacity versus the distance $r_F$ between the feed and the RIS, with the number of RIS elements $N_R=100$, transmit power $P=43$~dBm, the gain of the feed $\kappa=10$~dB, and cross-polarization coefficient $l_{RU}=0.2$.}
%	%	\vspace{-0.8cm}
%	\label{fig_capacity_vs_feed_distance}
%\end{figure}

%Fig.~\ref{fig_capacity_vs_feed_distance} shows how the system capacity changes with the distance $r_F$ between the RIS and the feed. Here, the upper bound is calculated based on (\ref{C_R_ub_max}) while the curve ``Monte Carlo" is acquired by averaging over $10000$ Monte Carlo results based on (\ref{C_R_ub_1}). According to Fig.~\ref{fig_capacity_vs_feed_distance}, we can find that the derived upper bound is consistent with the Monte Carlo results, which thus verifies the correctness of the theoretical analysis. Besides, we can find that as the feed moves away from the RIS, the ergodic capacity first improves and then shrinks, which is consistent with Remark~\ref{remark_the_alpha_distance_on_capacity}. In addition, we can observe that when the feed is sufficiently close to or sufficiently far from the RIS, the upper bound of the ergodic capacity is tight, which thus justifies the selection of the upper bound $C_{dual}^{ub}$ as the performance metric throughout this paper.
\begin{figure}[!tpb]
	\centering
	\vspace{0pt}
	\includegraphics[width=0.43\textwidth]{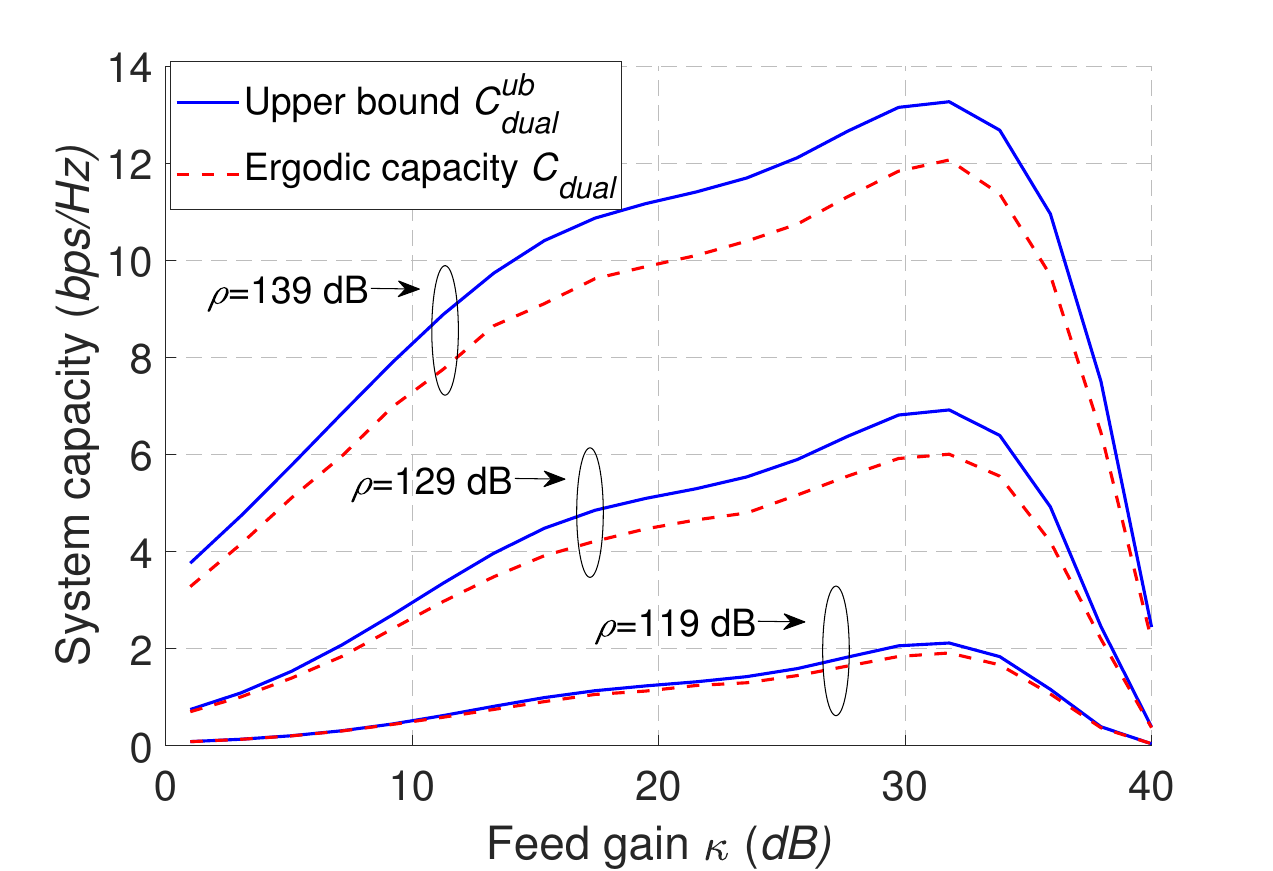}
	\vspace{-0.1cm}
	\caption{System capacity versus the gain $\kappa$ of the feed, with the number of RIS elements $N_R=100$ and cross-polarization coefficient $l_{RU}=0.2$. $\rho$ represents the ratio between transmit power and noise power.}
	%	\vspace{-0.7cm}
	\label{fig_capacity_vs_kappa}
\end{figure}

Fig.~\ref{fig_capacity_vs_kappa} depicts how the system capacity changes with the gain $\kappa$ of the feed. The transmit power is equally allocated between the two polarizations. According to Fig.~\ref{fig_capacity_vs_kappa}, it can be found as the feed gain becomes larger, both the ergodic capacity and its upper bound increase first and then degrade, which is consistent with Remark~\ref{remark_the_alpha_distance_on_capacity}. Moreover, we can observe that as the SNR $\rho$ tends to zero, the upper bound is asymptotically tight, which verifies Remark~\ref{remark_asymptotically_tight}. 

\begin{figure}[!tpb]
	\centering
	\vspace{0pt}
	\includegraphics[width=0.43\textwidth]{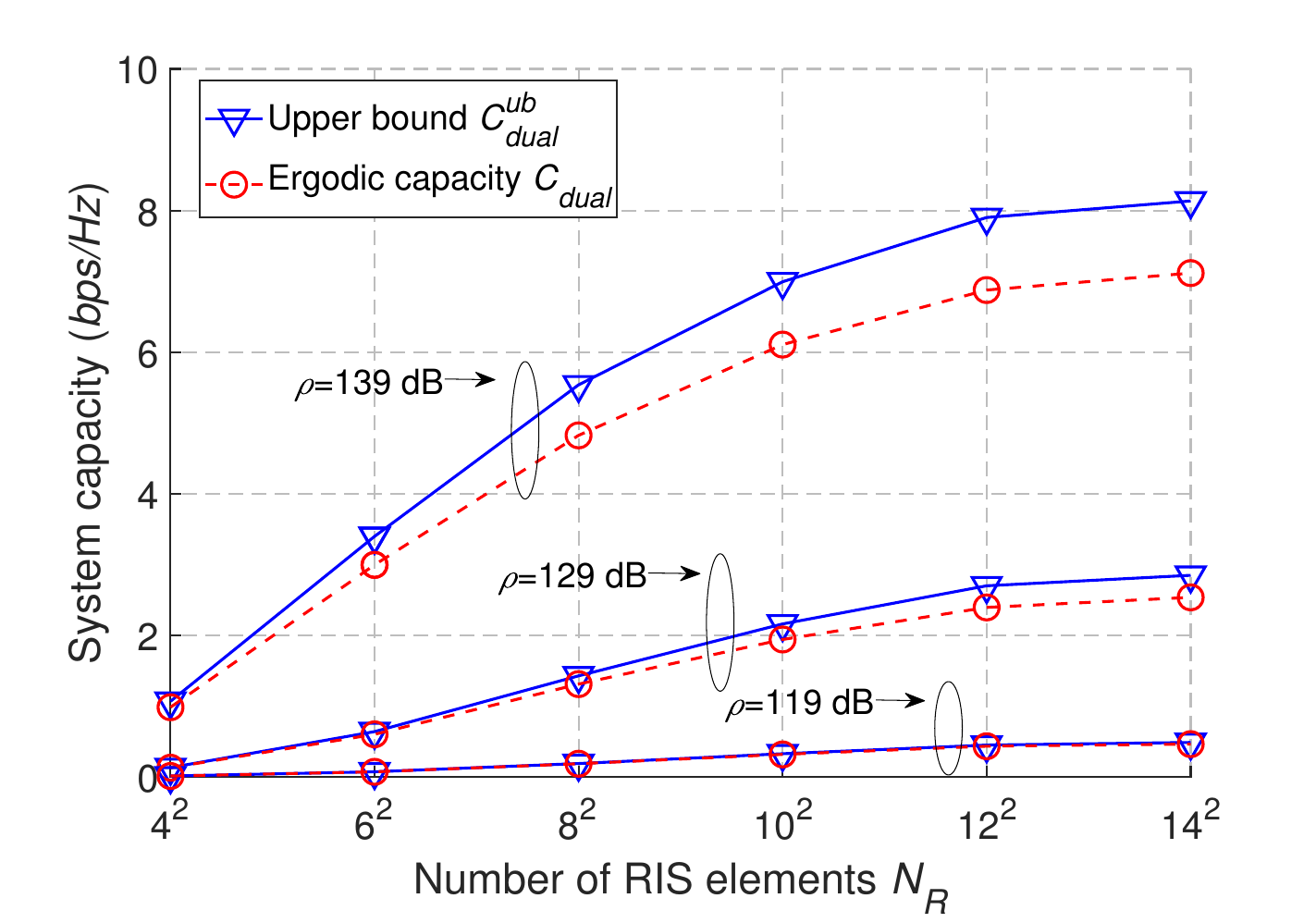}
	\vspace{-0.1cm}
	\caption{System capacity versus the number $N_R$ of RIS elements, with the gain of feed $\kappa=17$~dB and cross-polarization coefficient $l_{RU}=0.2$. $\rho$ represents the ratio between transmit power and noise power.}
	%	\vspace{-0.7cm}
	\label{fig_capacity_vs_element_num}
\end{figure}

%\begin{figure}[!tpb]
%	\centering
%	\vspace{0pt}
%	\includegraphics[width=0.43\textwidth]{capacity_vs_feed_gain_v2.eps}
%	\vspace{-0.1cm}
%	\caption{System capacity versus the gain $\kappa$ of the feed, with the number of RIS elements $N_R=100$ and cross-polarization coefficient $l_{RU}=0.2$. $\rho$ represents the ratio between transmit power and noise power.}
%%	\vspace{-0.7cm}
%	\label{fig_capacity_vs_element_num}
%\end{figure}

{Fig.~\ref{fig_capacity_vs_element_num} depicts how the system capacity changes with the number $N_R$ of RIS elements. According to Fig.~\ref{fig_capacity_vs_element_num}, it can be found that for all transmit SNR, both the ergodic capacity and its upper bound increases with the size of the RIS. This is because the RIS can capture more power radiated by the feed, thereby delivering more power to the UE through beamforming. However, as the number of RIS elements continues to increase, both the ergodic capacity and its upper bound eventually saturate. This is because most power radiated by the feed is captured by the RIS, and thus further increasing the size of the RIS has little impact on the power of the received desired signal. Moreover, we can observe that for any given number of RIS elements, the upper bound is asymptotically tight as the SNR $\rho$ tends to zero.}

%\begin{figure}[!tpb]
%	\centering
%	\vspace{0pt}
%	\includegraphics[width=0.43\textwidth]{capacity_vs_transmit_power_3_phase_shift_cmp_v3.eps}
%		\vspace{-0.1cm}
%	\caption{Performance comparison among different phase shift configuration schemes, with the number of RIS elements $N_R=100$, the gain of the feed $\kappa=10$~dB, and cross-polarization coefficient $l_{RU}=0.2$.}
%%		\vspace{-0.5cm}
%	\label{fig_capacity_vs_transmit_power_3_phase_shift_cmp_v2}
%\end{figure}
\begin{figure}[!tpb]
	\centering
	\vspace{0pt}
	\includegraphics[width=0.41\textwidth]{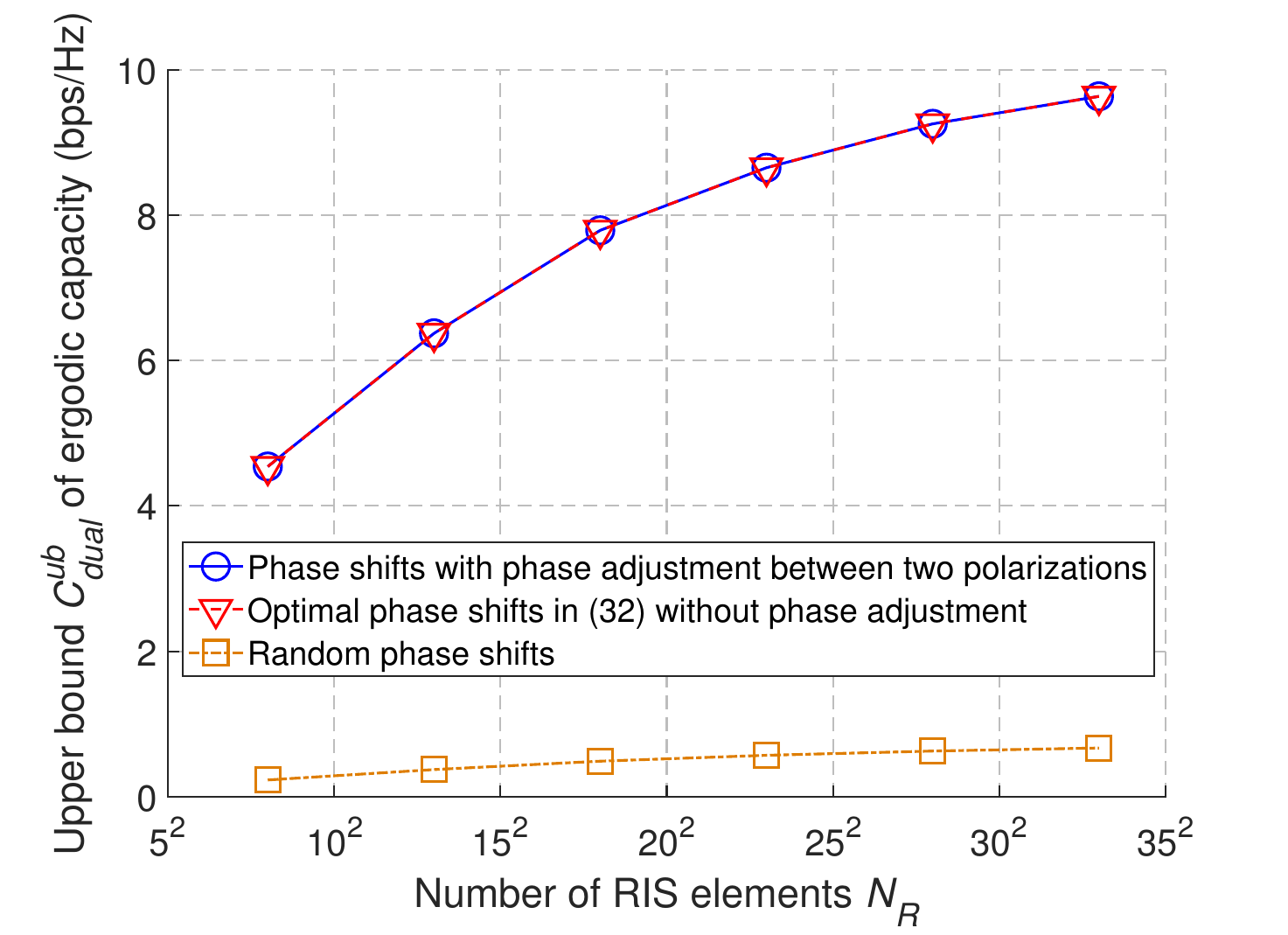}
	\vspace{-0.1cm}
	\caption{Performance comparison among different phase shift configuration schemes, with the transmit power $P=43$~dBm, the gain of the feed $\kappa=10$~dB, and cross-polarization coefficient $l_{RU}=0.2$.}
	%		\vspace{-0.5cm}
	\label{fig_capacity_vs_element_num_3_phase_shift_cmp_v2}
\end{figure}

In Fig.~\ref{fig_capacity_vs_element_num_3_phase_shift_cmp_v2}, we compare the system capacity $C_{dual}^{ub}$ achieved by three phase shift configuration schemes. The curve ``Phase shifts with phase adjustment between two polarizations" corresponds to the configuration $\varphi_n^{(V)}=\frac{2\pi}{\lambda}D_n-\varphi^{(VV)}$ and $\varphi_n^{(H)}=\frac{2\pi}{\lambda}D_n-\varphi^{(HH)}$, where the phase adjustment can be calculated as $\varphi_n^{(V)}-\varphi_n^{(H)}=\varphi^{(HH)}-\varphi^{(VV)}$. Further, the power allocation is given, where the power is equally allocated between two polarizations. From Fig.~\ref{fig_capacity_vs_element_num_3_phase_shift_cmp_v2}, we can find that the proposed phase shift configuration in (\ref{opt_phase}) outperforms the random scheme, which thus verifies the optimality of the proposed scheme. Moreover, it can be observed that the system capacity achieved by the optimal phase shifts without phase adjustment is equal to that when the phase adjustment is applied. This shows that the maximization of the system capacity does not require the phase adjustment between two polarizations, as indicated in Remark~\ref{remark_phase_adjustment}. {These results are consistent with those in Fig.~\ref{fig_capacity_vs_element_num}, and thus can also be explained in the same manner as that for Fig.~\ref{fig_capacity_vs_element_num}.}

\begin{figure}[!tpb]
	\centering
	\vspace{0pt}
	\includegraphics[width=0.43\textwidth]{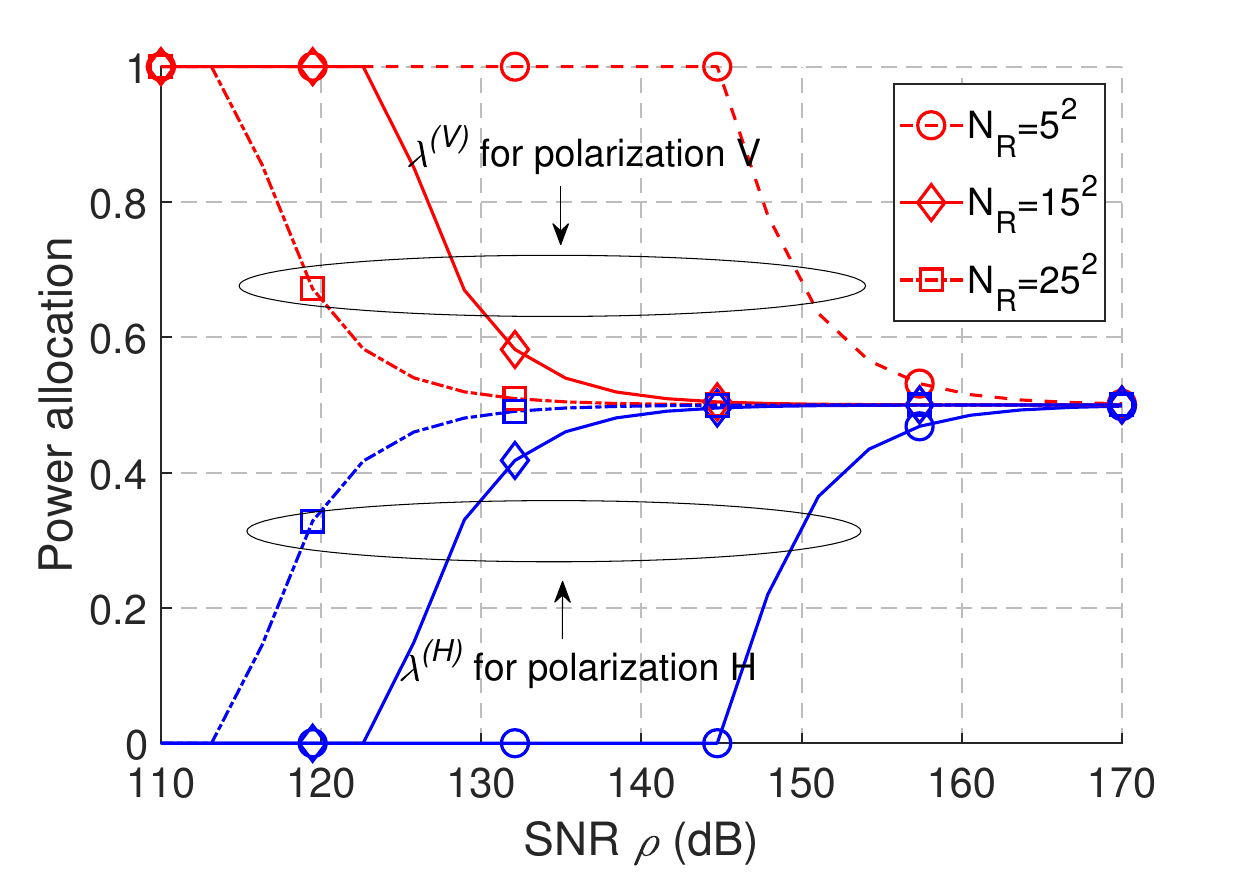}
		\vspace{-0.2cm}
	\caption{Power allocations across polarizations versus signal-to-noise ratio $\rho$, with the gain of the feed $\kappa=10$~dB, cross-polarization coefficient $l_{RU}=0.2$, the distance between the feed and the RIS $r_F=0.1$~m, and the elevation angle of the feed $\theta_F=\frac{\pi}{3}$. {Further, the average reflection amplitude of the RIS elements in the $V$-polarization is larger than that in the $H$-polarization, i.e., $\frac{1}{N_R}\sum_{n=1}^{N_R}A_n^{(V)}>\frac{1}{N_R}\sum_{n=1}^{N_R}A_n^{(H)}$.}}
%		\vspace{-0.5cm}
	\label{fig_power_allocation_rho}
\end{figure}

In Fig.~\ref{fig_power_allocation_rho}, we depict how the power allocations over polarizations change with SNR $\rho$. From Fig.~\ref{fig_power_allocation_rho}, we can find that unlike conventional dual-polarized MIMO where transmit poewr is equally allocated to different polarizations, more power is allocated to polarization $V$ than polarization $H$ in the dual-polarized RIS-aided system, i.e., $\lambda^{(V)}>\lambda^{(H)}$, since the reflection amplitude $\Gamma^{(V)}$ in polarization $V$ is larger than that in polarization $H$. Moreover, it can be observed that when the SNR $\rho$ increases or the RIS contains more elements, the gap between power allocations $\lambda^{(V)}$ and $\lambda^{(H)}$ becomes smaller, which is consistent with Remark~\ref{remark_rho_N}.

\begin{figure}[!tpb]
	\centering
	\vspace{0pt}
	\includegraphics[width=0.43\textwidth]{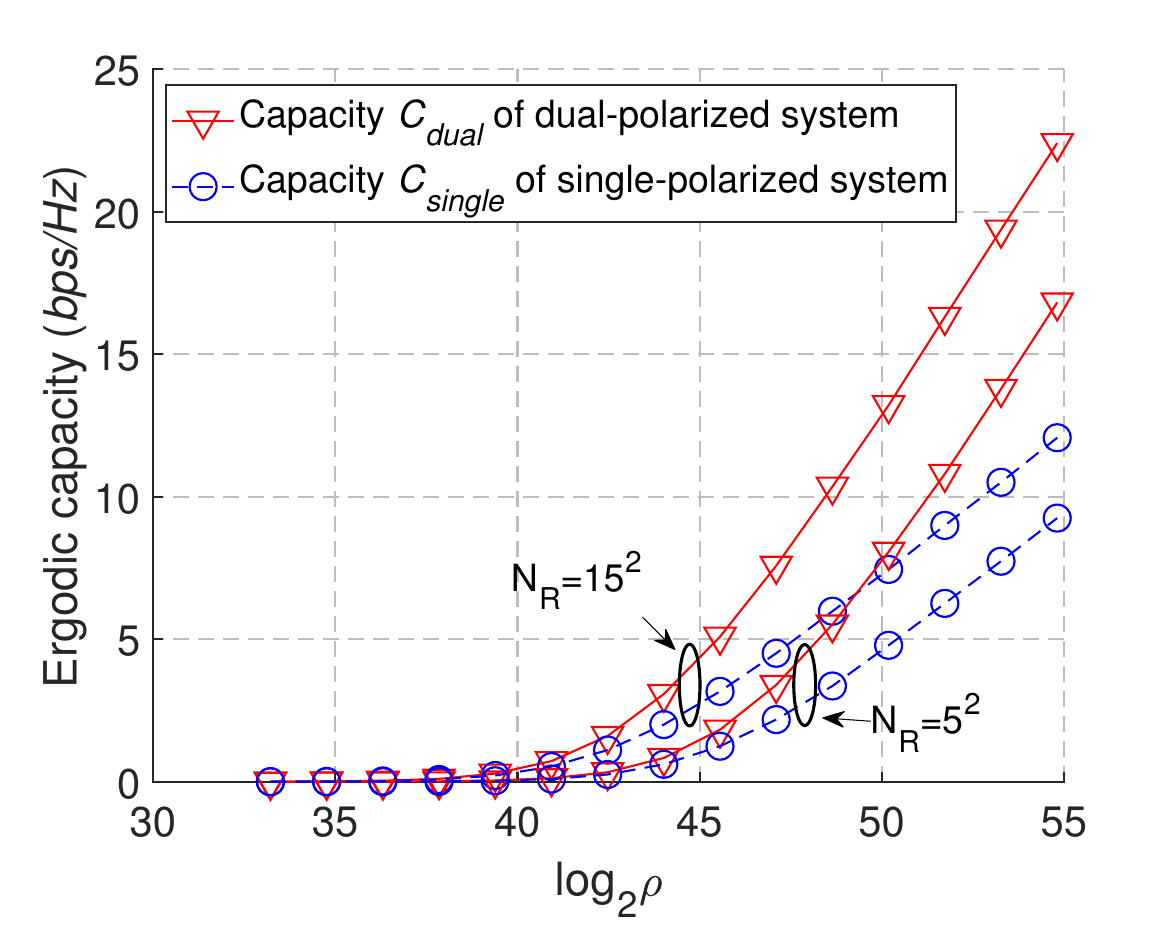}
	%	\vspace{-0.25cm}
	\caption{System capacity versus SNR $\rho$, with distance between the feed and the RIS $r_F=0.1$~m, and the gain of the feed $\kappa=10$~dB.}
	%	is set as $(r_F,\theta_F,\phi_F)=(0.1~m, \pi/3, \pi)$, and thus
	%	\vspace{-0.8cm}
	\label{fig_capacity_vs_rho_v3}
\end{figure}

Fig.~\ref{fig_capacity_vs_rho_v3} depicts how the ergodic capacity changes with the SNR $\rho$. According to Fig.~\ref{fig_capacity_vs_rho_v3}, the dual-polarized system achieves higher multiplexing gain, i.e., the growth rate of the ergodic capacity with respect to the SNR when the SNR is sufficiently large, is larger than that of the single-polarized system, than the single-polarized one, which is consistent with Theorem~\ref{theorem_multiplexing_gain_cmp}. In addition, we can observe that the ergodic capacity is positively correlated with the number of the RIS elements, since more power radiated by the feed can be captured by the RIS and then be delivered to the UE, leading to higher received SNR.

\begin{figure}[!tpb]
	\centering
	\subfigure[System capacity vs. feed deployment under feed gain $\kappa=10$~dB]{
		\begin{minipage}[b]{0.48\textwidth}
			\centering
			\includegraphics[width=1\textwidth]{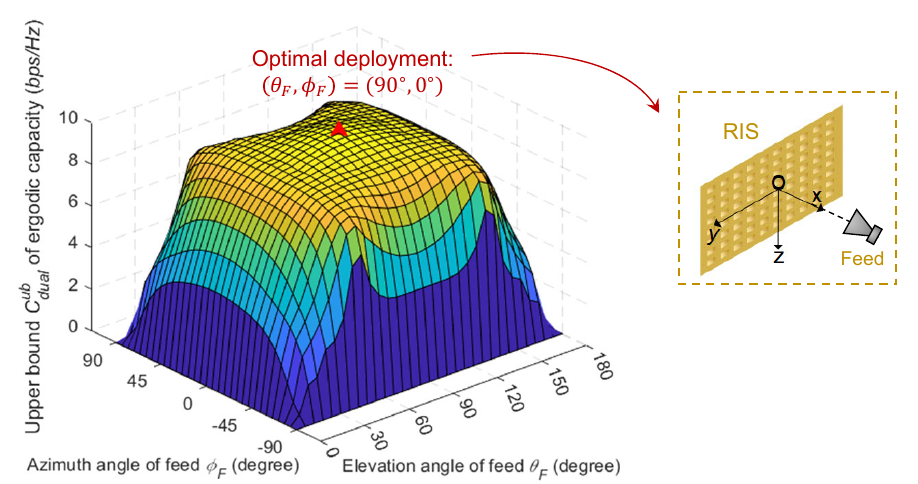}
			\vspace{-0.2cm}
			%		\caption{Amplitude vs. frequency}
			%		\vspace{-0.cm}
			\label{feed_deployment_10dB}
	\end{minipage}}

	\subfigure[System capacity vs. feed deployment under feed gain $\kappa=20$~dB]{
		\begin{minipage}[b]{0.48\textwidth}
			\centering
			\includegraphics[width=1\textwidth]{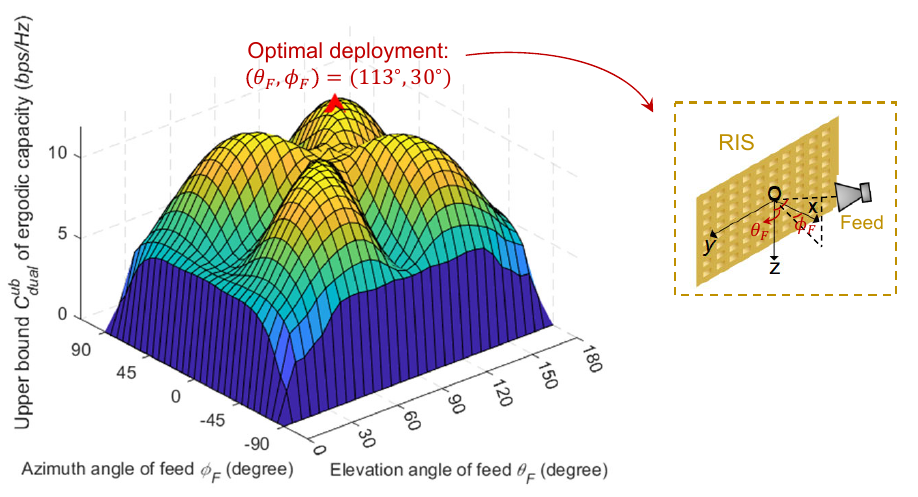}
			\vspace{-0.2cm}
			%		\caption{Phase response vs. frequency}
			%		\vspace{-0.3cm}
			\label{feed_deployment_20dB}
	\end{minipage}}
	\vspace{-0.2cm}
	\caption{System capacity $C_{dual}^{ub}$ of dual-polarized RIS-based system versus the azimuth angle $\phi_F$ and elevation angle $\theta_F$ of the feed. We set the number of RIS elements as $N_R=400$, transmit power as $P=43$~dBm, and cross-polarization coefficient as $l_{RU}=0.2$. Besides, the main lobe of the feed is directed towards the origin.}
	\label{feed_deployment}
%	\vspace{-0.5cm}
\end{figure}

Note that the azimuth angle $\phi_F$ and the elevation angle $\theta_F$ of the feed are correlated with the incident angles of the RIS elements, and thus they have an influence on the system performance. To this end, in Fig.~\ref{feed_deployment}, we investigate the relationship between the system capacity $C_{dual}^{ub}$ and the feed deployment. From Fig.~\ref{feed_deployment}, we can find that when the gain of the feed is relatively small, the feed should be deployed directly in front of the RIS, i.e., $(\theta_F,\phi_F)=(90^\circ,0^\circ)$, while when the gain of the feed is relatively large, it should be positioned above and at an angle to the RIS. 

In Fig.~\ref{fig_capacitu_vs_XPD_dual_vs_single_v4}, we plot the system capacity upper bound versus cross-polarization coefficient $l_{RU}$. We can observe from Fig.~\ref{fig_capacitu_vs_XPD_dual_vs_single_v4} that as $l_{RU}$ grows larger, the capacity upper bound $C_{dual}^{ub}$ of the dual-polarized system first increases and then decreases. Besides,  $C_{dual}^{ub}$ is maximized at $l_{RU}=1$ and $l_{RU}=0$ while it is minimized when $l_{RU}=\frac{1}{2}$, with the minimum capacity more than half of the maximum capacity, which verifies Theorem~\ref{theorem_XPD}. Unlike the dual-polarized system, the capacity upper bound $C_{single}^{ub}$ for the single-polarized system monotonically decreases with $l_{RU}$. Also, when the channel perfectly matches (i.e., $l_{RU}=0$) and absolutely mismatches (i.e., $l_{RU}=1$) the polarization of the single-polarized system, the capacity upper bound is maximized and minimized, respectively, with the minimum capacity upper bound equal to zero, which is consistent with Remark~\ref{remark_single_polarized_1} and Remark~\ref{remark_single_polarized_2}. Fig.~\ref{fig_capacitu_vs_XPD_dual_vs_single_v4} also shows that the capacity of the dual-polarized system is more than twice that of the single-polarized one when $l_{RU}$ exceeds the derived threshold, which verifies Remark~\ref{remark_single_polarized_3}. {From Fig.~\ref{fig_capacitu_vs_XPD_dual_vs_single_v4}, we can also observe that the ratio term $\frac{C_{R}^{ub}}{C_{single}^{ub}}$ is varying smoothly up to almost $l_{RU}=0.6$, and then increases sharply with $l_{RU}$. This is because as $l_{RU}\rightarrow 1$, the capacity upper bound $C_{single}^{ub}$ of the single-polarized system approaches $0$.}

\begin{figure}[!tpb]
	\centering
	\subfigure[Capacity upper bound vs. cross-polarization coefficient $l_{RU}$]{
		\begin{minipage}[b]{0.48\textwidth}
			\centering
			\includegraphics[width=.95\textwidth]{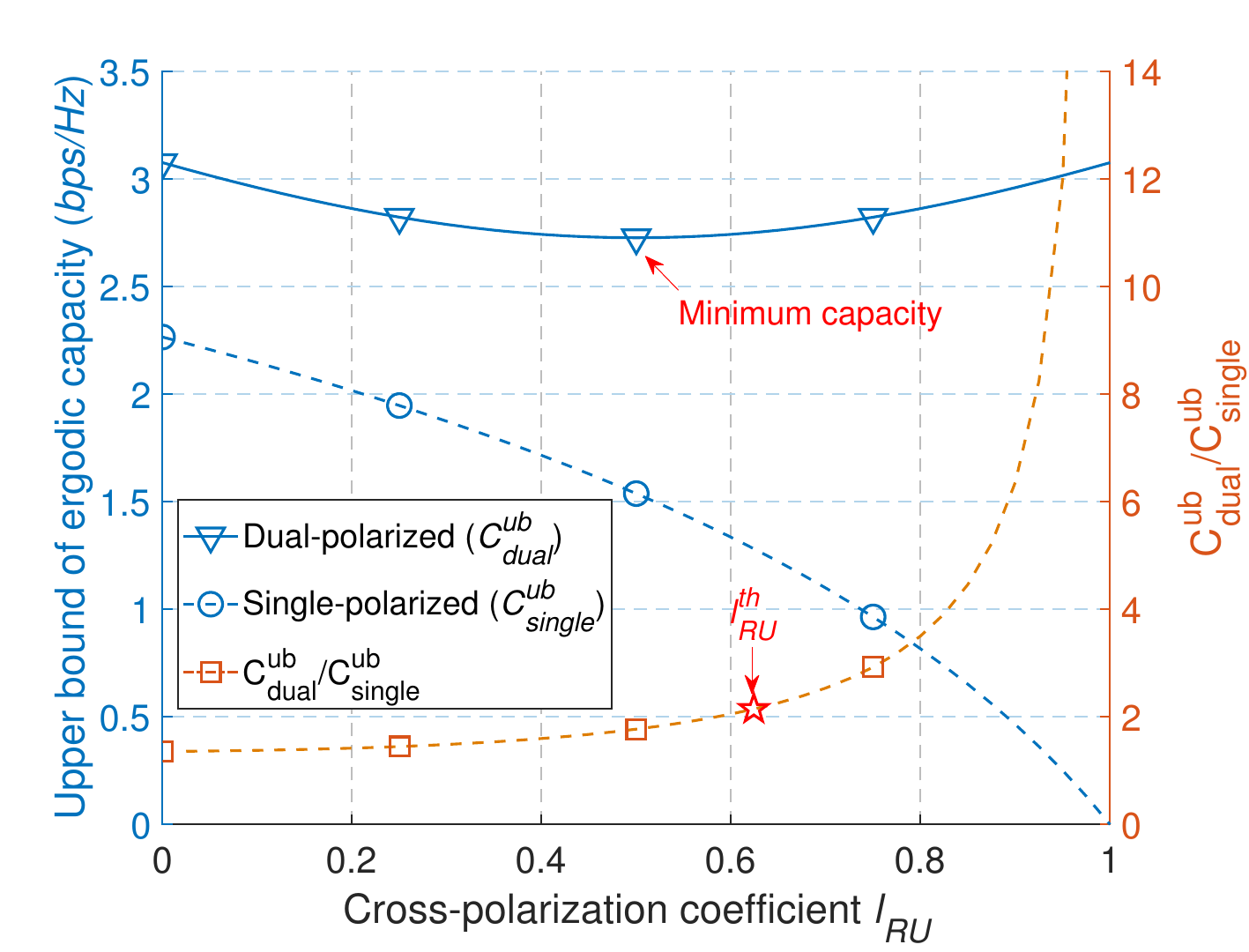}
			\vspace{-0.2cm}
			%		\caption{Amplitude vs. frequency}
			%		\vspace{-0.cm}
			\label{fig_capacitu_vs_XPD_dual_vs_single_v4}
	\end{minipage}}
	
	\subfigure[Ergodic capacity vs. cross-polarization coefficient $l_{RU}$]{
		\begin{minipage}[b]{0.48\textwidth}
			\centering
			\includegraphics[width=.95\textwidth]{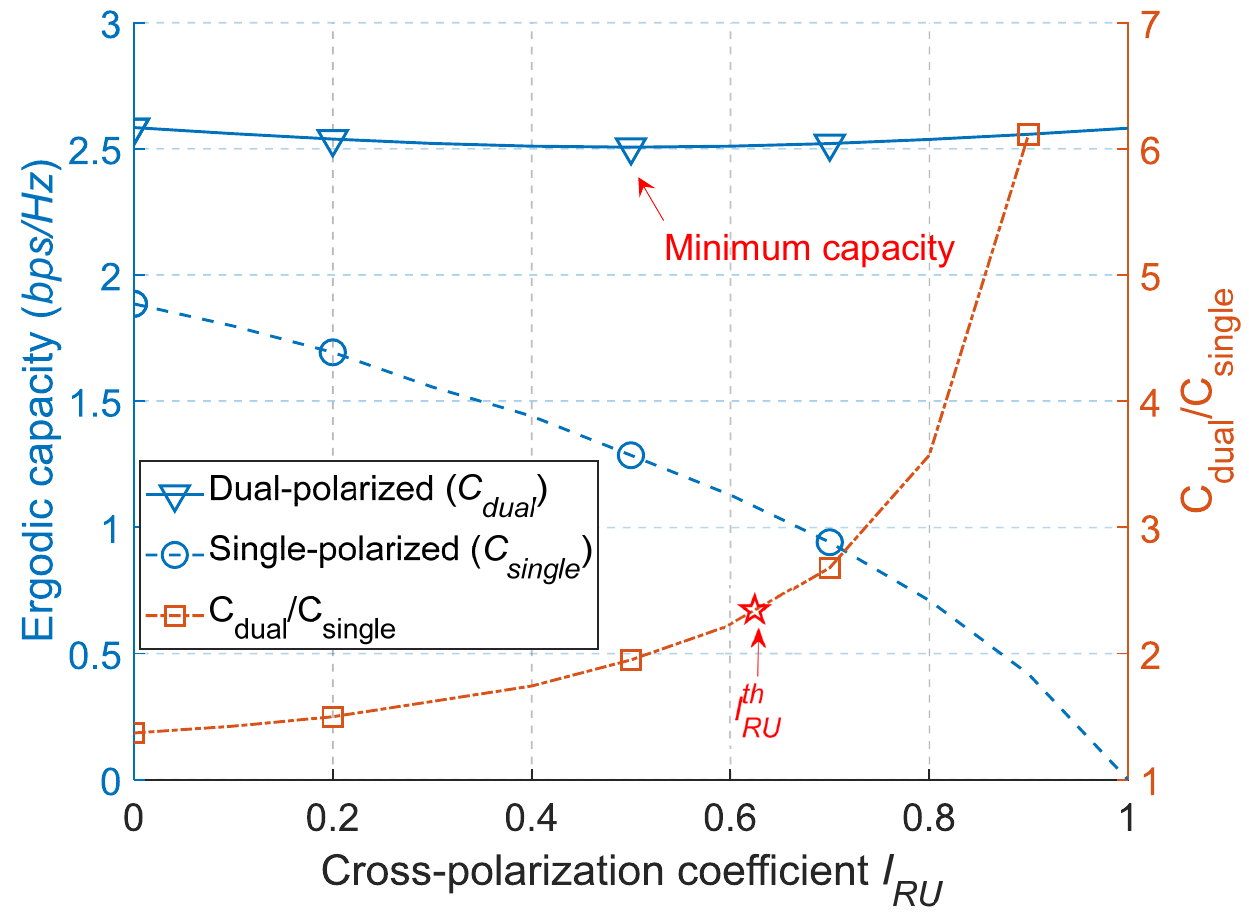}
			\vspace{-0.2cm}
			%		\caption{Phase response vs. frequency}
			%		\vspace{-0.3cm}
			\label{fig_ergodic_capacity_vs_XPD}
	\end{minipage}}
	\vspace{-0.2cm}
	\caption{Influences of cross-polarization coefficient $l_{RU}$ on system capacity and its upper bound, with the number of RIS elements $N_R=15^2$, transmit power $P=35$~dBm, distance between the feed and the RIS $r_F=0.1$~m, and the gain of the feed $\kappa=10$~dB. Note that both the RIS and the UE antenna are $V$-polarized. Therefore, when $l_{RU}=0$ and $l_{RU}=1$, the channel perfectly matches and absolutely mismatches the polarization of the single-polarized system, respectively.}
	\label{fig_capacity_vs_XPD}
	%	\vspace{-0.5cm}
\end{figure}

%\begin{figure}[!tpb]
%	\centering
%	\vspace{0pt}
%	\includegraphics[width=0.43\textwidth]{capacitu_vs_XPD_dual_vs_single_v7.eps}
%	%	\vspace{-0.25cm}
%	\caption{Capacity upper bound versus cross-polarization coefficient $l_{RU}$, with the number of RIS elements $N_R=15^2$, transmit power $P=35$~dBm, distance between the feed and the RIS $r_F=0.1$~m, the elevation angle of the feed $\theta_F=\frac{\pi}{3}$, and the gain of the feed $\kappa=10$~dB. Note that both the RIS and the UE antenna are $V$-polarized. Therefore, when $l_{RU}=0$ and $l_{RU}=1$, the channel perfectly matches and absolutely mismatches the polarization of the single-polarized system, respectively.}
%	%	is set as $(r_F,\theta_F,\phi_F)=(0.1~m, \pi/3, \pi)$, and thus
%	%	\vspace{-0.8cm}
%	\label{fig_capacitu_vs_XPD_dual_vs_single_v4}
%\end{figure}

{Fig.~\ref{fig_ergodic_capacity_vs_XPD} depicts the influences of cross-polarization coefficient $l_{RU}$ on the ergodic capacity. From Fig.~\ref{fig_ergodic_capacity_vs_XPD}, we can find that the trend of the ergodic capacity and the trend of the capacity ratio with respect to $l_{RU}$ are consistent with those corresponding to the capacity upper bound. {Further, both the ergodic capacity and its upper bound are minimized when $l_{RU}=\frac{1}{2}$. In addition, the derived threshold $l_{RU}^{th}$ is also effective for the ergodic capacity.} This demonstrates that the proposed upper bound can effectively characterize the ergodic capacity.}

%\begin{figure}[!tpb]
%	\centering
%	\vspace{0pt}
%	\includegraphics[width=0.43\textwidth]{ergodic_capacity_vs_XPD_v5.eps}
%	\vspace{-0.1cm}
%	\caption{Ergodic capacity versus cross-polarization coefficient $l_{RU}$, with the number of RIS elements $N_R=15^2$, transmit power $P=35$~dBm, distance between the feed and the RIS $r_F=0.1$~m, and the gain of the feed $\kappa=10$~dB. Note that both the RIS and the UE antenna are $V$-polarized. Therefore, when $l_{RU}=0$ and $l_{RU}=1$, the channel perfectly matches and absolutely mismatches the polarization of the single-polarized system, respectively.}
%	\vspace{-0.5cm}
%	\label{fig_ergodic_capacity_vs_XPD}
%\end{figure}
%\vspace{0.5cm}
\section{Conclusion}
\label{sec_conclusion}
In this paper, we have considered a dual-polarized RIS-enabled single-user HMIMO network. To facilitate analysis, an asymptotically tight upper bound on the ergodic capacity has been derived, based on which the power allocations across polarizations have been optimized. Given the optimal power allocations, the dual-polarized system has been compared against its single-polarized counterpart in terms of multiplexing gain and capacity upper bound. From our analytical and numerical results, the following main conclusions can be drawn:
\begin{itemize}
%	\item As the XPD of the dual-polarized RIS-based channel increases, i.e., the received co-polarization signal becomes stronger while the cross-polarization one becomes weaker, the system capacity first decreases and then increases. The maximum system capacity is achieved when XPD is equal to $0$ or $\infty$, i.e., no co-polarization component or no cross-polarization component, while the capacity is minimized when XPD is equal to $1$, i.e., the received co-polarization and cross-polarization signals have the same strength.
	\item Unlike conventional dual-polarized massive MIMO with uniform power allocations across polarizations, the BS should allocate more power to the polarization direction corresponding to larger average reflection amplitudes of the RIS elements in a dual-polarized RIS-aided system.
%	This variation arises from differences in reflection amplitudes associated with different polarizations due to the angle-dependence of RIS elements. 
	\item The capacity of the dual-polarized system is more than twice that of the single-polarized one when the channel mismatches the polarization of the single-polarized system, i.e., when the polarization of the received signal is nearly orthogonal to that of the UE antenna. Besides, a closed-form threshold for polarization parameter XPD can be derived to achieve such capacity gain.
	
	\item System capacity is influenced by the deployment of the feed through incident angles of RIS elements. To maximize system capacity, the feed should be deployed directly in front of the RIS when the feed gain is relatively small. However, under relatively large feed gain, the feed should be positioned above and at an angle to the RIS.
	
%	The deployment of the feed is coupled with the feed gain, i.e., when the gain of the feed is relatively small, the feed should be deployed directly in front of the RIS in order to maximize system capacity while when the gain of the feed is relatively large, the feed should be positioned above and at an angle to the RIS.
\end{itemize}

\begin{appendices}
%\vspace{0.5cm}
\section{Proof of Theorem~\ref{theorem_upper_bound}}
\label{appendix_upper_bound}
According to Jensen's inequality, i.e.,
\begin{align}
	\mathbb{E}\{\log_2(x)\}\le \log_2(\mathbb{E}\{x\}),
\end{align}
 the ergodic capacity in (\ref{ergodic_capacity}) is upper bounded by,
\begin{align}
\label{upper_bound}
C_{dual}\le\log_2\left(\mathbb{E}\{\det(\bm{I}_2+\rho\bm{G}\bm{\Lambda}\bm{G}^{\dagger})\}\right).
%\triangleq C_R^{ub}.
\end{align}

For the simplicity of exposition, we define
\begin{align}
\label{def_A}
\bm{A}=\bm{G}\bm{\Lambda}\bm{G}^{\dagger}=[A_{1,1},A_{1,2};A_{2,1},A_{2,2}].
\end{align}
% $\bm{A}=\bm{G}\bm{G}^{\dagger}=[A_{1,1},A_{1,2};A_{2,1},A_{2,2}]$. 
 To this end, the entries of the matrix $\bm{A}$ can be written as a function of the entries of $\bm{G}$, i.e.,
\begin{align}
\label{A_11}
A_{1,1}=\lambda^{(V)}|G_{1,1}|^2+\lambda^{(H)}|G_{1,2}|^2,\\
\label{A_12}
A_{1,2}=\lambda^{(V)}G_{1,1}G_{2,1}^*+\lambda^{(H)}G_{1,2}G_{2,2}^*,\\
\label{A_21}
A_{2,1}=\lambda^{(V)}G_{2,1}G_{1,1}^*+\lambda^{(H)}G_{2,2}G_{1,2}^*,\\
\label{A_22}
A_{2,2}=\lambda^{(V)}|G_{2,1}|^2+\lambda^{(H)}|G_{2,2}|^2.
\end{align}	
By substituting (\ref{def_A}) into the expression $C_{dual}^{ub}$ in (\ref{upper_bound}), we have
\begin{align}
\label{derivation_C_R_ub}
C_{dual}^{ub}&=\log_2\left(\mathbb{E}\left\{\det(\bm{I}_2+\rho\bm{A})\right\}\right)\notag\\
&=\log_2\left(\mathbb{E}\{(1+\rho A_{1,1})(1+\rho A_{2,2})-\rho^2 A_{2,1}A_{1,2}\}\right)\notag\\
&=\log_2\Big(1+\rho\mathbb{E}\{A_{1,1}\}+\rho\mathbb{E}\{A_{2,2}\}+\rho^2\mathbb{E}\{A_{1,1}A_{2,2}\}\notag\\
&\quad\quad\quad-\rho^2\mathbb{E}\{A_{2,1}A_{1,2}\}\Big).
\end{align}
To this end, we will focus on the derivations of $\mathbb{E}\{A_{1,1}\}$, $\mathbb{E}\{A_{2,2}\}$, $\mathbb{E}\{A_{1,1}A_{2,2}\}$, and $\mathbb{E}\{A_{1,2}A_{2,1}\}$ in the following. 

To be specific, according to (\ref{A_11}) and (\ref{A_22}), we have
\begin{align}
\label{derivation_E_A11}
\mathbb{E}\{A_{1,1}\}=\lambda^{(V)}\mathbb{E}\{|G_{1,1}|^2\}+\lambda^{(H)}\mathbb{E}\{|G_{1,2}|^2\},\\
\label{derivation_E_A22}
\mathbb{E}\{A_{2,2}\}=\lambda^{(V)}\mathbb{E}\{|G_{2,1}|^2\}+\lambda^{(H)}\mathbb{E}\{|G_{2,2}|^2\}.
\end{align}
According to (\ref{A_11}), (\ref{A_22}) and (\ref{G_11})-(\ref{G_22}), we can see that $A_{1,1}$ and $A_{2,2}$ depend only on $(\bm{h}^{(VV)},\bm{h}^{(VH)})$ and $(\bm{h}^{(HV)},\bm{h}^{(HH)})$, respectively. Recall that the four polarization components of the RIS-based channel matrix $\bm{H}$, i.e., $\bm{h}^{(VV)}$, $\bm{h}^{(VH)}$, $\bm{h}^{(HV)}$, and $\bm{h}^{(HH)}$ are independent. Therefore, $A_{1,1}$ and $A_{2,2}$ are also independent variables, from which we can obtain the following formula:
\begin{align}
\label{derivation_A11_A22}
\mathbb{E}\{A_{1,1}A_{2,2}\}=\mathbb{E}\{A_{1,1}\}\mathbb{E}\{A_{2,2}\}.
\end{align}

Finally, we move on to the derivation of $\mathbb{E}(A_{1,2}A_{2,1})$. Based on (\ref{A_12}) and (\ref{A_21}), we have
\begin{align}
\label{derivation_A12_A21}
&\mathbb{E}\{A_{1,2}A_{2,1}\}\notag\\
&=(\lambda^{(V)})^2\mathbb{E}\{|G_{1,1}|^2|G_{2,1}|^2\}+(\lambda^{(H)})^2\mathbb{E}\{|G_{1,2}|^2|G_{2,2}|^2\}\notag\\
&\quad+\lambda^{(V)}\lambda^{(H)}\mathbb{E}\{G_{1,1}(G_{1,2})^*(G_{2,1})^*G_{2,2}\}\notag\\
&\quad+\lambda^{(V)}\lambda^{(H)}\mathbb{E}\{(G_{1,1})^*G_{1,2}G_{2,1}(G_{2,2})^*\}.
\end{align}
Similarly, due to the independence of different polarization components of the RIS-aided channel $\bm{H}$, the entries of the equivalent channel matrix $G$, i.e., $G_{1,1}$, $G_{1,2}$, $G_{2,1}$ and $G_{2,2}$, are also independent based on (\ref{G_11})-(\ref{G_22}). Therefore, we have
\begin{align}
\label{derivation_E_G_11_module_square_G_21_module_square}
\mathbb{E}\{|G_{1,1}|^2|G_{2,1}|^2\}=\mathbb{E}\{|G_{1,1}|^2\}\mathbb{E}\{|G_{2,1}|^2\},
\end{align}
\begin{align}
\label{derivation_E_G_12_module_square_G_22_module_square}
\mathbb{E}\{|G_{1,2}|^2|G_{2,2}|^2\}=\mathbb{E}\{|G_{1,2}|^2\}\mathbb{E}\{|G_{2,2}|^2\},
\end{align}
\begin{align}
&\mathbb{E}\{G_{1,1}(G_{1,2})^*(G_{2,1})^*G_{2,2}\}\notag\\
&=\mathbb{E}\{G_{1,1}\}\mathbb{E}\{G_{1,2}\}^*\mathbb{E}\{G_{2,1}\}^*\mathbb{E}\{G_{2,2}\},
\end{align}
%\begin{align}
%&\mathbb{E}(G_{1,1}(G_{1,2})^*(G_{2,1})^*G_{2,2})\notag\\
%&=\mathbb{E}(G_{1,1})\mathbb{E}((G_{1,2})^*)\mathbb{E}((G_{2,1})^*)\mathbb{E}(G_{2,2}),
%\end{align}
%\begin{align}
%&\mathbb{E}((G_{1,1})^*G_{1,2}G_{2,1}(G_{2,2})^*)\notag\\
%&=\mathbb{E}((G_{1,1})^*)\mathbb{E}(G_{1,2})\mathbb{E}(G_{2,1})\mathbb{E}((G_{2,2})^*),
%\end{align}
%\begin{align}
%\mathbb{E}(\!(G_{1,1})^*G_{1,2}G_{2,1}(G_{2,2})^*\!)\!=\!\mathbb{E}(G_{1,1})^*\mathbb{E}(G_{1,2})\mathbb{E}(G_{2,1})\mathbb{E}(G_{2,2})^*\!,
%\end{align}
\begin{align}
&\mathbb{E}\{(G_{1,1})^*G_{1,2}G_{2,1}(G_{2,2})^*\}\notag\\
&=\mathbb{E}\{G_{1,1}\}^*\mathbb{E}\{G_{1,2}\}\mathbb{E}\{G_{2,1}\}\mathbb{E}\{G_{2,2}\}^*\!.
\end{align}
Since the RIS-aided channels $h_n^{ji}$ are Rayleigh faded, we have
\begin{align}
\label{derivation_E_G_11_G_12_G_21_G_22}
\mathbb{E}\{G_{1,1}(G_{1,2})^*(G_{2,1})^*G_{2,2}\}=0,\\
\label{derivation_E_G_11_G_12_G_21_G_22_conjugate}
\mathbb{E}\{(G_{1,1})^*G_{1,2}G_{2,1}(G_{2,2})^*\}=0.
\end{align}
By substituting (\ref{derivation_E_G_11_module_square_G_21_module_square}), (\ref{derivation_E_G_12_module_square_G_22_module_square}), (\ref{derivation_E_G_11_G_12_G_21_G_22}), and (\ref{derivation_E_G_11_G_12_G_21_G_22_conjugate}) into (\ref{derivation_A12_A21}), we can obtain that
\begin{align}
\label{derivation_A12_A21_final}
\mathbb{E}\{A_{1,2}A_{2,1}\}=&(\lambda^{(V)})^2\mathbb{E}\{|G_{1,1}|^2\}\mathbb{E}\{|G_{2,1}|^2\}\notag\\
&+(\lambda^{(H)})^2\mathbb{E}\{|G_{1,2}|^2\}\mathbb{E}\{|G_{2,2}|^2\}.
\end{align}
By combining (\ref{derivation_C_R_ub}), (\ref{derivation_E_A11}), (\ref{derivation_E_A22}), (\ref{derivation_A11_A22}), (\ref{derivation_A12_A21_final}), we can prove (\ref{C_R_ub_1}).
%\vspace{0.5cm}
\section{Proof of Remark~\ref{remark_asymptotically_tight}}
\label{appendix_asymptotically_tight}
To simplify the presentation, we define
\begin{align}
	a=\mathbb{E}\left\{|G_{1,1}|^2+|G_{1,2}|^2+|G_{2,1}|^2+|G_{2,2}|^2\right\},\\
	b=\mathbb{E}\{|G_{1,1}|^2\}\mathbb{E}\{|G_{2,2}|^2\}+\mathbb{E}\{|G_{1,2}|^2\}\mathbb{E}\{|G_{2,1}|^2\}.
\end{align}
Then, the capacity upper bound in (\ref{upper_bound}) can be rewritten as
\begin{align}
\label{reformulation_capacity_upper_bound}
C_{dual}^{ub}=\log_2\left(1+\frac{a}{2}\rho+\frac{b}{4}\rho^2\right).
\end{align}

Recall that $\log_2(1+x)\le\frac{x}{\ln2}$. Therefore, based on (\ref{reformulation_capacity_upper_bound}), we have
\begin{align}
\label{upper_bound_upper_bound}
C_{dual}^{ub}\le \frac{a}{2\ln 2}\rho+\frac{b}{4\ln 2}\rho^2.
\end{align}

A lower bound on the ergodic capacity $C_{dual}$ is also derived. According to~\cite{Oyman_statistical_2003}, we have
\begin{align}
\label{instant_rate_lower_bound}
\log_2\det\left(\bm{I}_2+\frac{\rho}{2}\bm{G}\bm{G}^{\dagger}\right)\ge \log_2\left(1+\frac{\rho}{2}\mathrm{Tr}(\bm{G}\bm{G}^{\dagger})\right).
\end{align}
Note that 
\begin{align}
	\log_2\left(1+x\right)\ge \frac{1}{\ln 2}\left(x-\frac{x^2}{2}\right).
\end{align}
Therefore, we have
\begin{align}
\label{instant_rate_lower_bound_v2}
\log_2\left(1+\frac{\rho}{2}\mathrm{Tr}(\bm{G}\bm{G}^{\dagger})\right)\ge \frac{\mathrm{Tr}(\bm{G}\bm{G}^{\dagger})}{2\ln 2}\rho-\frac{(\mathrm{Tr}(\bm{G}\bm{G}^{\dagger}))^2}{8\ln 2} \rho^2.
\end{align}
By substituting (\ref{instant_rate_lower_bound}) and (\ref{instant_rate_lower_bound_v2}) into the expression of the ergodic capacity in (\ref{ergodic_capacity}), a lower bound on the ergodic capacity can be derived as follows:
\begin{align}
\label{ergodic_capacity_lower_bound}
C_{dual}&=\mathbb{E}\{\log_2\det(\bm{I}_2+\frac{\rho}{2}\bm{G}\bm{G}^{\dagger})\}\notag\\
&\ge \frac{\mathbb{E}\{\mathrm{Tr}(\bm{G}\bm{G}^{\dagger})\}}{2\ln 2}\rho-\frac{\mathbb{E}\{\mathrm{Tr}(\bm{G}\bm{G}^{\dagger})\}^2}{8\ln 2} \rho^2.
\end{align}

By combining (\ref{upper_bound_upper_bound}) and  (\ref{ergodic_capacity_lower_bound}), we have
\begin{align}
C_{dual}^{ub}-C_{dual}\le f(\rho),
\end{align}
where $f(\rho)$ is given by
\begin{align}
f(\rho)=&\left(\frac{a}{2\ln 2}-\frac{\mathbb{E}\{\mathrm{Tr}(\bm{G}\bm{G}^{\dagger})\}}{2\ln 2}\right)\rho\notag\\
&+\left(\frac{b}{4\ln 2}+\frac{\mathbb{E}\{\mathrm{Tr}(\bm{G}\bm{G}^{\dagger})\}^2}{8\ln 2}\right)\rho^2.
\end{align}
Note that
\begin{align}
	\lim_{\rho\rightarrow 0}f(\rho)=0.
\end{align}
Also, we have $C_{dual}^{ub}-C_{dual}\ge 0$. Therefore, according to the Squeeze Theorem, we have
\begin{align}
\lim_{\rho\rightarrow 0} C_{dual}^{ub}-C_{dual}=0,
\end{align} 
which completes the proof.

%\section{Proof of Remark~\ref{remark_tighter_parameter}}
%\label{appendix_tighter_parameter}
%\vspace{0.5cm}
\section{Proof of Theorem~\ref{the_opt_phase}}
\label{app_opt_phase}
According to~(\ref{C_R_ub_1}), we can find that the capacity upper bound $C_{dual}^{ub}$ achieves its maximum when $\mathbb{E}\{|G_{1,1}|^2\}$, $\mathbb{E}\{|G_{1,2}|^2\}$, $\mathbb{E}\{|G_{2,1}|^2\}$, and $\mathbb{E}\{|G_{2,2}|^2\}$ are simultaneously maximized. To this end, we first focus on the maximization of $\mathbb{E}\{|G_{1,1}|^2\}$. 

According to (\ref{G_11}), we have
%By uncovering the expression of $G_{1,1}$
%\begin{align}
%	\mathbb{E}(|G_{1,1}|^2)&=\mathbb{E}(|\bm{h}^{(VV)}\bm{\Gamma}^{(V)}\bm{b}^{(VV)}|^2)\notag\\
%	&=\mathbb{E}((\bm{h}^{(VV)}\bm{\Gamma}^{(V)}\bm{b}^{(VV)})(\bm{h}^{(VV)}\bm{\Gamma}^{(V)}\bm{b}^{(VV)})^*)\notag\\
%	&=\mathbb{E}(\sum_nh_n^{})
%\end{align}
\begin{align}
	\label{derivation_E_G_11_module_square}
	&\mathbb{E}\{|G_{1,1}|^2\}\notag\\
	&=\mathbb{E}\left\{|\bm{h}^{(VV)}\bm{\Gamma}^{(V)}\bm{b}^{(VV)}|^2\right\},\notag\\
	&=\mathbb{E}\left\{(\bm{h}^{(VV)}\bm{\Gamma}^{(V)}\bm{b}^{(VV)})(\bm{h}^{(VV)}\bm{\Gamma}^{(V)}\bm{b}^{(VV)})^*\right\},\notag\\
	&=\mathbb{E}\left\{(\sum_n h_n^{(VV)}\Gamma_n^{(V)}b_n^{(VV)})(\sum_n h_n^{(VV)}\Gamma_n^{(V)}b_n^{(VV)})^*\right\}\notag\\
	&=\mathbb{E}\left\{\sum_{n_1}\sum_{n_2} h_{n_1}^{(VV)}(h_{n_2}^{(VV)})^*\Gamma_{n_1}^{(V)}(\Gamma_{n_2}^{(V)})^*b_{n_1}^{(VV)}(b_{n_2}^{(VV)})^*\right\}\notag\\
	&=\sum_{n_1}\sum_{n_2}\Gamma_{n_1}^{(V)}(\Gamma_{n_2}^{(V)})^*b_{n_1}^{(VV)}(b_{n_2}^{(VV)})^* \mathbb{E}\left\{h_{n_1}^{(VV)}(h_{n_2}^{(VV)})^*\right\}.
\end{align}
According to the RIS-aided channel model in (\ref{model_RRS_2_UE}) and the definition of the correlation matrix $\bm{R}$ in (\ref{def_correlation_matrix}), we have
\begin{align}
	\label{correlation_RRS_aided_channel}
	\mathbb{E}\{h_{n_1}^{(VV)}(h_{n_2}^{(VV)})^*\}=\sqrt{\beta_{n_1}^{(VV)}\beta_{n_2}^{(VV)}}\bm{R}(n_1,n_2).
\end{align}
By substituting (\ref{correlation_RRS_aided_channel}) into (\ref{derivation_E_G_11_module_square}), $\mathbb{E}(|G_{1,1}|^2)$ can be further written as
\begin{align}
	\label{derivation_E_G_11_module_square_v2}
%	\mathbb{E}(|G_{1,1}|^2)&=\sum_{n_1}\sum_{n_2}\Gamma_{n_1}^{(V)}(\Gamma_{n_2}^{(V)})^*b_{n_1}^{(VV)}(b_{n_2}^{(VV)})^* \sqrt{\beta_{n_1}^{(VV)}\beta_{n_2}^{(VV)}}\bm{R}(n_1,n_2).
&\mathbb{E}\{|G_{1,1}|^2\}\notag\\
&=\sum_{n_1, n_2}\!\Gamma_{n_1}^{(V)}b_{n_1}^{(VV)}(\Gamma_{n_2}^{(V)}b_{n_2}^{(VV)})^* \sqrt{\beta_{n_1}^{(VV)}\beta_{n_2}^{(VV)}}\bm{R}(n_1,n_2).
\end{align}

When the reflection phase shifts are set to those given in (\ref{opt_phase}), $\mathbb{E}\{|G_{1,1}|^2\}$ can be rewritten as
\begin{align}
	\label{derivation_E_G_11_module_square_max}
	\mathbb{E}\{|G_{1,1}|^2\}=(1-l_{RU})O^{(V)}.
\end{align}
Similarly, we can prove that given the reflection phase shifts in (\ref{opt_phase}), $\mathbb{E}\{|G_{1,2}|^2\}$, $\mathbb{E}\{|G_{2,1}|^2\}$, and $\mathbb{E}\{|G_{2,2}|^2\}$ are also maximized as
\begin{align}
	\label{derivation_E_G_12_module_square_max}
	\mathbb{E}\{|G_{1,2}|^2\}=l_{RU}O^{(H)},\\
	\label{derivation_E_G_21_module_square_max}
	\mathbb{E}\{|G_{2,1}|^2\}=l_{RU}O^{(V)},\\
	\label{derivation_E_G_22_module_square_max}
	\mathbb{E}\{|G_{2,2}|^2\}=(1-l_{RU})O^{(H)},
\end{align}
respectively. By substituting (\ref{derivation_E_G_11_module_square_max}), (\ref{derivation_E_G_12_module_square_max}), (\ref{derivation_E_G_21_module_square_max}), and (\ref{derivation_E_G_22_module_square_max}) into (\ref{C_R_ub_1}), we can obtain the maximized capacity upper bound under the optimal reflection phase shifts as shown in (\ref{C_R_ub_max}), which completes the proof.
%\vspace{0.5cm}
\section{Proof of Theorem~\ref{theorem_XPD}}
\label{appendix_XPD} 
It is easy to prove that $C_{dual}^{ub}$ increases with $l_{RU}^2+(1-l_{RU})^2$. Therefore, to find out how the capacity upper bound $C_{dual}^{ub}$ changes with the cross-polarization coefficient $l_{RU}$, we need to consider the trend of $l_{RU}^2+(1-l_{RU})^2$. To this end, we define 
\begin{align}
	g(l_{RU})=l_{RU}^2+(1-l_{RU})^2,
\end{align}
whose derivative is given by
%According to (\ref{C_R_ub_max}), the derivative of capacity upper bound $C_{dual}^{ub}$ can be given by
\begin{align}
	\label{derivative_upper_bound}
	g'(l_{RU})=4l_{RU}-2.
\end{align}
It follows that $g'(l_{RU})$ is negative when $l_{RU}<\frac{1}{2}$ while it is non-negative otherwise, which completes the proof. 

\section{Proof of Lemma~\ref{lemma_single_polarization}}
\label{appendix_single_polarization}
The received signal of a user with a $V$-polarized antenna from the $V$-polarized feed through the $V$-polarized RIS can be given by
\begin{align}
\label{rec_sig_single}
y_{single}=\bm{h}^{(VV)}\bm{\Gamma}^{(V)}\bm{b}^{(VV)}x^{(V)}+n,
\end{align}
where $\bm{b}^{(VV)}$ is the propagation effect from the feed to the RIS defined in (\ref{propagation_coeff}), $\bm{\Gamma}^{(V)}$ represents the matrix consisting of the reflection coefficients of the RIS elements defined in (\ref{Gamma}), $\bm{h}^{(VV)}$ denotes the channel from the RIS to the UE given in (\ref{channel}), $x^{(V)}$ is the transmitted signal of the $V$-polarized feed, and $n$ is AWGN. According to (\ref{rec_sig_single}), the ergodic capacity of the single-polarized RIS-based system can be obtained, as given by (\ref{ergodic_capacity_single}). Similarly, by applying Jensen's inequality, an upper bound on the system capacity can be derived, as shown in (\ref{ergodic_capacity_single_upper_bound}). Similar to the proof of Theorem~\ref{the_opt_phase}, it is easy to find that by setting the phase shifts of the RIS as in (\ref{opt_phase_single}), the capacity upper bound on the single-polarized system will be maximized, with the maximum capacity upper bound given by (\ref{ergodic_capacity_single_upper_bound_opt_phase}). 
%expressed as
%\begin{align}
%C_{single}=\mathbb{E}(\log_2(1+\rho|G_{1,1}|^2)),
%\end{align}
%where the equality holds when 
%\vspace{0.5cm}
\section{Proof of Theorem~\ref{theorem_multiplexing_gain_cmp}}
\label{appendix_multiplexing_gain_cmp}
Recall the definition of the matrix $\bm{A}$ in (\ref{A_11})-(\ref{A_22}), based on which we have
\begin{align}
\label{instant_capacity}
&\log_2\det(\bm{I}_2+\frac{\rho}{2}\bm{G}\bm{G}^{\dagger})\notag\\
=&\log_2\left(1+\frac{\rho}{2}(A_{1,1}+A_{2,2})+\frac{\rho^2}{4}(A_{1,1}A_{2,2}+A_{2,1}A_{1,2})\right).
\end{align}
Note that the entries of $\bm{A}$, i.e., $A_{1,1}$, $A_{1,2}$, $A_{2,1}$, and $A_{2,2}$ are not correlated with the transmit SNR $\rho$. Therefore, based on (\ref{instant_capacity}), we can obtain
\begin{align}
\label{instant_capacity_multiplexing_gain}
\lim_{\rho\rightarrow +\infty} \frac{\log_2\det(\bm{I}_2+\frac{\rho}{2}\bm{G}\bm{G}^{\dagger})}{\log_2\rho}=2.
\end{align}
By combining (\ref{instant_capacity_multiplexing_gain}) and (\ref{ergodic_capacity}), we have
\begin{align}
\lim_{\rho\rightarrow +\infty}\frac{C_{dual}}{\log_2\rho}&=\lim_{\rho\rightarrow +\infty}\frac{\mathbb{E}\left\{\log_2\det(\bm{I}_2+\frac{\rho}{2}\bm{G}\bm{G}^{\dagger})\right\}}{\log_2\rho}\notag\\
&=\mathbb{E}\left\{\lim_{\rho\rightarrow +\infty}\frac{\log_2\det(\bm{I}_2+\frac{\rho}{2}\bm{G}\bm{G}^{\dagger})}{\log_2\rho}\right\}\notag\\
&=2.
\end{align}
Similarly, we can prove (\ref{multiplexing_gain_single}).
\begin{figure*}[!hb]
	%	{\noindent}
	\rule[-12pt]{17.5cm}{0.05em}
	\begin{equation}
		\setlength{\abovedisplayskip}{3pt}
		\setlength{\belowdisplayskip}{-3pt}
		\begin{aligned}
			\label{lower_bound_of_upper_bound}
			C_{dual}^{ub} \ge \log_2\left(1+\frac{\rho(O^{(H)}+O^{(V)})}{2}+ \frac{\rho^2O^{(H)}O^{(V)}(l_{RU}^2+(1-l_{RU})^2)}{4}\right)
			\triangleq \hat{C}_{dual}^{ub},
		\end{aligned}
	\end{equation}
\end{figure*}
%\vspace{0.5cm}
\section{Proof of Remark~\ref{remark_single_polarized_3}}
\label{app_single_polarized_3}
First, a lower bound on the capacity $C_{dual}^{ub}$ of the dual-polarized RIS-based system is derived. Note that $C_{dual}^{ub}$ in (\ref{C_R_ub_max_max}) corresponds to the maximum system capacity under an optimal power allocation $((\lambda^{(H)})^*,(\lambda^{(V)})^*)$. Therefore, it is lower bounded by the system capacity under equal power allocations over different polarizations, as indicated in (\ref{lower_bound_of_upper_bound}) shown at the bottom of the next page. 

Then, we focus on proving that when $l_{RU}>l_{RU}^{(th)}$, the lower bound given in (\ref{lower_bound_of_upper_bound}) is more than twice that of the capacity of the single-polarized RIS-based system, i.e., $\hat{C}_{dual}^{ub}-2C_{single}^{ub}>0$, based on which we can prove Remark~\ref{remark_single_polarized_3}. Specifically, due to the monotonicity of $\log(x)$, we can obtain (\ref{equivalent_double_capacity}) shown at the bottom of the next page.
\begin{figure*}[!hb]
	%	{\noindent}
	\rule[-12pt]{17.5cm}{0.05em}
	\begin{equation}
		\setlength{\abovedisplayskip}{3pt}
		\setlength{\belowdisplayskip}{-3pt}
		\begin{aligned}
			\label{equivalent_double_capacity}
			\hat{C}_{dual}^{ub}-2C_{single}^{ub}>0\iff 1+\frac{\rho(O^{(H)}+O^{(V)})}{2}+ \frac{\rho^2O^{(H)}O^{(V)}(l_{RU}^2+(1-l_{RU})^2)}{4}-\left(1+\rho(1-l_{RU})O^{(V)}\right)^2>0.
		\end{aligned}
	\end{equation}
\end{figure*} 
Define $g(l_{RU})$ as the function describing the relationship between the left-hand-side of the inequality in (\ref{equivalent_double_capacity}) and $l_{RU}$, where $g(l_{RU})$ is given by
\begin{align}
	g(l_{RU})=a(l_{RU})^2+bl_{RU}+c.
\end{align} 
%\begin{align}
%	g(l_{RU})=a(l_{RU})^2+bl_{RU}+c.
%\end{align}
Note that $l_{RU}^{(th)}$ is the smaller of the two roots of $g(l_{RU})$. Further, we have 
\begin{align}
	a<0,\quad-\frac{b}{2a}>0,\quad g(0)<0,\quad g(1)>0
\end{align}
%and
%\begin{align}
%	.
%\end{align}
%$a<0$, $-\frac{b}{2a}>0$, $g(0)<0$, and $g(1)>0$. 
Therefore, when $l_{RU}$ exceeds the threshold $l_{RU}^{(th)}$, $g(l_{RU})$ always takes positive values, which completes the proof.
\end{appendices}

\end{document}